\documentclass[10pt,journal,compsoc]{IEEEtran}
\ifCLASSOPTIONcompsoc
  \usepackage[nocompress]{cite}
\else
  \usepackage{cite}
\fi
\ifCLASSINFOpdf
\else
\fi

\IEEEoverridecommandlockouts
\usepackage{graphicx}
\usepackage{color}
\usepackage{placeins}
\usepackage{float}
\usepackage{tabularx,colortbl}
\usepackage{subcaption}
\usepackage{nccmath}

\usepackage[latin1]{inputenc}
\usepackage{amsfonts}
\usepackage{amssymb}
\usepackage{epstopdf}
\usepackage[ruled,vlined]{algorithm2e}
\usepackage{dsfont}

\usepackage{amsthm}

\theoremstyle{definition}
\newtheorem{definition}{Definition}[section]
\newtheorem{theorem}{Theorem}[section]
\newtheorem{corollary}{Corollary}[section]
\newtheorem{lemma}{Lemma}[section]

\begin{document}
%
\title{Estimating Parameters of Large CTMP from Single Trajectory with Application to Stochastic Network Epidemics Models}
%
%
%
%

\author{Seyyed A. Fatemi,~\IEEEmembership{Member,~IEEE,}
        June Zhang,~\IEEEmembership{Member,~IEEE}
\IEEEcompsocitemizethanks{\IEEEcompsocthanksitem J. Zhang is with the Department
of Electrical and Computer Engineering, University of Hawai'i at M\={a}noa, Honolulu, HI, 96822.\protect\\
E-mail: zjz@hawaii.edu
}
\thanks{Manuscript received April 19, 2005; revised August 26, 2015.}}

%
%

\markboth{Journal of \LaTeX\ Class Files,~Vol.~14, No.~8, August~2015}%
{Shell \MakeLowercase{\textit{et al.}}: Bare Demo of IEEEtran.cls for Computer Society Journals}
%



\IEEEtitleabstractindextext{%
\begin{abstract}
Graph dynamical systems (GDS) model dynamic processes on a (static) graph. Stochastic GDS has been used for network-based epidemics models such as the contact process and the reversible contact process. In this paper, we consider stochastic GDS that are also continuous-time Markov processes (CTMP), whose transition rates are linear functions of some dynamics parameters $\theta$ of interest (i.e., healing, exogeneous, and endogeneous infection rates). Our goal is to estimate $\theta$ from a single, finite-time, continuously observed trajectory of the CTMP. Parameter estimation of CTMP is challenging when the state space is large; for GDS, the number of Markov states are \emph{exponential} in the number of nodes of the graph. We showed that holding classes (i.e., Markov states with the same holding time distribution) give efficient partitions of the state space of GDS. We derived an upperbound on the number of holding classes for the contact process, which is polynomial in the number of nodes. We utilized holding classes to solve a smaller system of linear equations to find $\theta$. Experimental results show that finding reasonable results can be achieved even for short trajectories, particularly for the contact process. In fact, trajectory length does not significantly affect estimation error.


\end{abstract}

}

\maketitle

\IEEEdisplaynontitleabstractindextext

%
\IEEEpeerreviewmaketitle

\IEEEraisesectionheading{\section{Introduction}\label{sec:introduction}}

%
%
%
%
Graphs are used to represent dependencies between variables. There has been a great deal of interest in graph-based data analysis. Graph/network-based dynamical systems are less well-studied and used in application. Graph dynamical systems (GDS) model processes on graphs. The behavior of these systems are determined by: 1) a graph, $G$, which characterizes the dependencies amongst variables (i.e., nodes), and 2) dynamics rules that describe how the variable values change over time. Often, the state of a variable is assumed to be explicitly dependent on the states of its neighboring nodes in $G$. In this paper, we will focus on systems where $G$ has a \emph{finite} number of nodes and is \emph{static} over time. Previously, graph dynamical systems over infinitely-size graph has been studied as either cellular automata or interacting particle systems, depending on if the system has discrete-time or continuous-time dynamics.

In stochastic GDS, the dynamic rules have a random component. This means that the future state of the system is characterized by a probability distribution. Stochastic GDS that are used to model the spread of epidemics over contact networks are known as \emph{stochastic network epidemics models}.

The Markov process is a natural mathematical tool for studying such dynamics. However, they are considered to be impractical for modeling GDS because the size of the state space, $D$, is exponential in the number of variables. Additionally, $D(D-1)$ parameters are needed to completely specify the Markov process. When the state space is very large, obtaining a sufficient number of observations to compute the likelihood for $D(D-1)$ parameters would require an infeasibly long observation window.

However, GDS-Markov processes generally have more structure than an arbitrary Markov process.  In many models, the pairwise transition probabilities/rates are functions of only a few variables. For example, the contact process is one of the most well-known GDS-continuous-time Markov process model \cite{harris1974contact}. It models infections between infected and healthy agents (i.e., nodes) in a known contact network. Only three parameters characterize the continuous-time dynamics of the model: 1) healing rate $\mu$, 2) exogeneous infection rate $\beta$ (i.e., infection from outside the contact network), 3) endogeneous infection rate $\delta$ (i.e., infection from contact with infected neighbor). Furthermore, the pairwise transition rates of the contact process are linearly dependent on $\mu, \beta, \delta$. We show that this induces equivalence classes in the diagonal entries of the transition rate matrix. These equivalence classes, which we call \emph{holding classes}, effectively reduces the dimensionality of the Markov process. This makes parameter estimation from a moderate-length trajectory possible. 

Leveraging the existence of holding classes and efficient least squares computation, we present an algorithm for learning the dynamics parameter, $\theta$, of GDS-continuous-time Markov process (CTMP) models whose transition rates are \emph{linearly} dependent on $\theta$. In this paper, we assume that the system is continuously observable (i.e., no missing observation) and that the graph $G$ is known. We applied our algorithm to the contact process and the reversible contact process. 

In Section~\ref{sec:contact} and~\ref{sec:rev}, we derived upperbounds on the number of holding classes for the contact process and the reversible contact process on arbitrary graph $G$. For the contact process, can be seen that this upperbound is polynomial, instead of exponential, in the number of variables. Therefore, the number of holding classes relative to the dimension of the state space \emph{decreases} with increasing $G$; we confirm this experimentally by exhaustively counting the number number of holding classes for small Erdos-Renyi and Watts-Strongatz graphs in Section~\ref{sec:numholdingexp}. Section~\ref{sec:exp} uses our algorithm in numerical experiments to estimate $\mu, \beta, \delta$ from synthetically generated trajectories on both 100-node random graphs and a 100-bus test power grid.

%
%
%

\section{Related Works}
We will review some of the current works on graph dynamical systems for the case where the underlaying network is \emph{static}. Related works can be found in various areas. Like the contact process, many such graph dynamical systems models are meant to model contagion-like behavior. As a result, there are a lot of overlap between works in epidemiology and general work on graph dynamical systems. We also review existing literatures on parameter estimation problems of continuous-time Markov processes; these approaches are intended for arbitrary Markov process and have not been applied to GDS due to the prohibitively large state space.

\subsection{Network-based Epidemics models and Social Contagion Models}
Efforts for mathematical modeling of contagion goes back to 17th century and Daniel Bernoulli's seminal work on smallpox disease spread. Typically, compartmental models are used, which segment individuals in a population by status such as susceptible ($S$, healthy), exposed ($E$, not infectious), infected ($I$, exposed and infectious), removed, etc. Epidemics models can be deterministic or stochastic \cite{kermack1927contribution,anderson1992infectious, dietz2002daniel}. 

Deterministic epidemics models predict the number/fraction of infected individuals in a population at some future time $t$ whereas stochastic models infer the probability distribution of the quantity of interest. Deterministic models are derived using mean-field approximation, which assume full-mixing (i.e every individual interacts with all other individuals with the same probability) and analyze the asymptotic (assuming an infinitely large population) behavior stochastic models. Classic epidemics models are time-series based; there are many approaches to learning the dynamics parameters of classic epidemics models \cite{Pan2014,Zheng2017, o2002tutorial, kypraios2017tutorial, Dutta2018}. 

One challenge is that the infection process are often not directly observable, and the likelihood is difficult to compute. Bayesian methods, which treats the unknown dynamics parameters and any unobserved quantities as random variables, are often favored. These approaches can be combined with with simple simulation models to infer the parameters that best fit the observed data. Reference~\cite{Dutta2018} uses approximate Bayesian computing (ABC) to estimate the dynamics parameter, $\theta$, and the single infection source of for a discrete-time Susceptible-Infected (SI) epidemics process.

Recently, greater attention has been paid to the inclusion of a heterogeneous contact networks within the epidemics models. These models are known as network epidemics models \cite{Pastor-Satorras2015a, Pellis2015}. Deterministic network-epidemics models extend mean-field results by incorporating statistics of the underlying network into the dynamics equations \cite{gleeson2013binary}. For example, the pairwise approximation model characterizes the dynamic of the number of contacts between an infected and susceptible individual,
\[
\frac{d [SI]}{dt} = -\tau[SI] + \tau[SSI] - \tau[ISI] + g[II] -g[SI],
\]
where $\tau$ and $g$ are the infection rate and healing rate respectively \cite{doi:10.1098/rsif.2005.0051}.

Stochastic network epidemics models are usually discrete-time (DTMP) or continuous-time (CTMP) Markov processes. A state in the Markov process is a possible configuration of the population. In discrete-time Markov model, the probability that a healthy node will be infected in the next time step is assumed to depend on how many infected neighbors it has in $G$ \cite{gomez2010discrete,Wang2017}. In continuous-time Markov model, the probability that a healthy node will be infected in a time interval $\tau$ is dependent on how many infected neighbors it has in $G$ \cite{Zhang2014,Zhang2015,Zhang2017,VanMieghem2012}. Stochastic network epidemics models have been studied in probability as interacting particle systems (IPS). The most well-known IPS model is the contact process \cite{harris1974contact, liggett1985interacting, griffeath1983binary}. These models also share many similarities with spin system models in physics \cite{glauber1963time}.

Research has mostly focused on the analysis of network epidemics models such as the impact of network topology on the epidemic threshold \cite{Wang2017, nowzari2016analysis}. Typically, these analytical results are derived for $t \to \infty$ and/or for an infinitely large network. Fewer work have focused on learning network epidemics models from data. The inclusion of $G$ induces a very large (discrete) state space, which makes parameter estimation difficult even if the infection process can be directly observed.

\subsection{Other Graph Dynamical System Models}

In social science, networked-based models of social contagion have been studied \cite{Centola2007}. These work differ from epidemics-based models in the contagion mechanism. For example, social contagion models often assume a directed contact network $G$. The edge direction implies that one person has more impact to affect the behavior/state of the other individual. Social contagion model also differentiate between simple and complex contagion. In simple contagion, a susceptible agent becomes infected with a fixed probability as a result of a one-time exposure to an infected neighbor. In complex contagion, a susceptible agent becomes infected with a probability that scales with the number of exposure to infected neighbors. 

In graph signal processing, the interest is in analyzing and filtering graph signal (real or complex variables associated with nodes in a static graph). Graph dynamical systems can be studied from the perspective of time-vertex graph signals \cite{loukas2019stationary, grassi2017time}. Various stationary properties are used to characterize the covariance between the nodal values in the graph and over time.

Dynamic graph models have also been studied as (time-varying) probabilistic graphical models. Reference~\cite{nodelman2012continuous} and~\cite{el2011continuous} introduced the continuous-time Bayesian network and continuous-time Markov network for representing dynamic directed (Bayesian) or undirected (Markov) probabilistic graphical model. Similar to works in network epidemics/graph dynamical systems, the dynamic of a variable (node in a network) is only dependent on the state of neighboring nodes. The CTMP model proposed by reference~\cite{el2011continuous}, called continuous-time Markov network (CTMN) studied, considered \emph{reversible} CTMPs only. A maximum likelihood estimator for learning the model parameters was presented for a binary-state CTMN. However, this was studied on a 4-node network, and the state space is only $D = 2^4$.

\subsection{Learning Continuous-time Markov Process}

Many work in parameter estimation of large-scale continuous-time Markov processes come from chemistry, where Markov State Models (MSM) are used to model molecular dynamics \cite{mcgibbon2015efficient, MSMpaper}. In these problems, only discrete-time observations are available. Therefore, a transition probability matrix, $P$, is learned from data and the used to estimate the transition rate matrix, $Q$. Large state-spaces are difficult to handle. If the methods do not assume nor exploit any additional structure in the Markov process, the runtime is on the order of $O(D^3)$ to $O(D^5)$, where $D$ is the number of Markov states. Reference~\cite{el2011continuous} proposed a more efficient maximum likelihood estimator of the transition rate matrix by exploiting the additional structure of \emph{reversible} continuous-time Markov processes.


\section{Continuous-time Markov Process}\label{sec:contmarkov}

The dynamic of a continuous-time Markov process, ${X(t), t\ge 0}$, taking values in a countable state space $\mathcal{D} = \{1,2, \ldots, D\}$, is completely described by the transition rate matrix, $Q =[q(i,j)], i,j \in \mathcal{D}$ (also known as the infinitesimal generator) \cite{norris1998markov, Kelly}. The transition rate from Markov state $i$ to state $j$ is
\[
q(i,j) = \lim_{\tau \to 0} \frac{P(X(t+\tau) = j | X(t) = i)}{\tau}, i \neq j.
\]
The diagonal entries of $Q$, $q(i,i), i \in \mathcal{D}$, are such that 
\begin{equation}\label{eq:diagonal}
\sum_{j \in \mathcal{D}} q(i,j) = 0.
\end{equation}
Therefore, 
\begin{equation}\label{eq:diagonal2}
q(i,i) = -\sum_{j \in \mathcal{D}} q(i,j).
\end{equation}
The continuous-time Markov process, $X(t)$, remains in state $i$ for a length of time that is exponentially distributed with rate $|q(i,i)|$; this is known as the \emph{holding time}. To reduce notation, whenever we refer to a diagonal entry $q(i,i)$, we mean the absolute value. 

\begin{definition}\label{def:holdingclass}
Two Markov states, $i, j \in \mathcal{D}, i \neq j$ belongs to the same \emph{holding class}, $\mathcal{H}$, if $q(i,i) = q(j,j)$. Then $q(\mathcal{H}, \mathcal{H}) = q(i,i) = q(j,j)$.
\end{definition}

Equivalently, two Markov states belong to the same holding class if their holding time has the same distribution. The state space of $X(t)$ can be partitioned into different holding classes: $\mathcal{D} = \{\mathcal{H}_1 \cup \mathcal{H}_2 \cup \ldots  \cup\mathcal{H}_K| \mathcal{H}_i \cap \mathcal{H}_j = \emptyset \}$.

\subsection{Continuous-time Markov Process with Linearly Dependent Transition Rates}

In this paper, we consider a sub-class of continuous-time Markov processes whose transition rates are linear functions depending on some underlying dynamic parameter,  $\theta \in \mathds{R}^b$, where $b << D$. This means that~\eqref{eq:diagonal2} can expressed in vector form as $F\theta$, where $F$ is some $D \times b$ matrix determined by how the transitions depend on $\theta$. However, we only need to be concerned with the unique diagonal values of $Q$ (i.e., holding classes), then
\begin{equation}\label{eq:linearsystem}
F\theta = \begin{bmatrix} 
q(\mathcal{H}_1,\mathcal{H}_1)\\ 
\vdots\\ 
q(\mathcal{H}_K,\mathcal{H}_K) 
\end{bmatrix}.
\end{equation}
As a result, the matrix $F$ is $K \times b$. Solving this system of equation is much more efficient when the number of holding classes is much smaller than the number of Markov states (i.e,. $K << D$). We will show that this is the case for network epidemics models such as the contact process. However, in order to solve the system of linear equations to find $\theta$, we need to be able to estimate the holding class rates, $q(\mathcal{H}_i, \mathcal{H}_i)$, from observation.

\subsection{Estimating Transitions Rates from Trajectory}


We assume that the states are continuously observed without noise. This means that 1) we do not miss any observations, 2) we know both the sampling time $t_i$ and the state of the system at $t_i$. Given a finite duration ($T= t_{M-1} - t_0$) trajectory
\[
\Sigma = \{x(t_0), x(t_1), x(t_2), \ldots, x(t_{M-1})\},
\]
it is known that the maximum likelihood estimator (MLE) of transition rate $q(i,j)$ is 
\begin{equation}\label{eq:MLEtransitionrate}
\widehat{q}(i,j)_{\text{MLE}}= \frac{N_{ij}(T)}{R_i(T)},
\end{equation}
where $N_{ij}(T)$ is the number of transitions from state $i \in \mathcal{D}$ to state $j \in \mathcal{D}$ in interval $T$ and
\[
R_i(T) = \int_{t_0}^{t_{M-1}} \mathds{1}(X(t) = i) dt, \quad i \in \mathcal{D}
\]
is the total amount of time $X(t)$ was in state $i$. The function $\mathds{1}(\cdot)$ is the indicator function. The MLE for the holding time rate is
\begin{equation}\label{eq:MLEdiagonaltransitionrate}
|\widehat{q}(i,i)_{\text{MLE}}| = \frac{\sum_{i\neq j }N_{ij}(T)}{R_i(T)}.
\end{equation}

Instead of considering individual Markov states, we can consider the holding classes. The MLE estimate of $q(\mathcal{H}_i,\mathcal{H}_i)$ is
\begin{equation}\label{eq:MLEdiagonaltransitionrate2}
\widehat{q}(\mathcal{H}_i, \mathcal{H}_i)_{\text{MLE}} =  \frac{\sum_{\mathcal{H}_i \neq \mathcal{H}_j} N_{\mathcal{H}_i, \mathcal{H}_j}(T)}{R_{\mathcal{H}_i}(T)},
\end{equation}
where $N_{\mathcal{H}_i, \mathcal{H}_j}(T)$ is the number of transitions from Markov states in holding class $\mathcal{H}_i$ to states in holding class $\mathcal{H}_j$ in $T$ intervals, and $R_{\mathcal{H}_i}(T)$ is the total amount of time the process remained in Markov states belonging to holding class $\mathcal{H}_i$.

It is known however, that the MLE of the exponential distribution is a biased estimate. An alternative to the MLE estimate is the uniformly minimum variance unbiased estimator (UMVUE):
\begin{equation}\label{eq:configdiag}
\widehat{q}(\mathcal{H}_i, \mathcal{H}_i)_{\text{UMVUE}} =  \frac{\sum_{\mathcal{H}_i \neq \mathcal{H}_j} N_{\mathcal{H}_i, \mathcal{H}_j}(T) - 1}{R_{\mathcal{H}_i}(T)}.
\end{equation}
Other estimators include the minimum mean squared error estimator (MMSE). Reference~\ref{cohen1973estimation} has more reviews of the different estimators of the exponential distribution.

\subsection{Finding $\widehat{\theta}$}

Depending on the duration of the observed trajectory, not all transitions between pairs of holding classes will be observed. Therefore, it may be that not $\widehat{q}(\mathcal{H}_1, \mathcal{H}_1), \ldots \widehat{q}(\mathcal{H}_K, \mathcal{H}_K)$ can be estimated. Therefore~\eqref{eq:linearsystem} may be overdetermined or underdetermined. We can account for the differences in estimation quality of each term by using the inverse sample variance. We can estimate $\theta$ by solving either the weighted least squares problem 
\begin{equation}\label{eq:wls}
\widehat{\theta} = \arg \min_{\theta} ||W^{\frac{1}{2}}(F\theta - [\widehat{q}(\mathcal{H}_1, \mathcal{H}_1), \ldots, \widehat{q}(\mathcal{H}_K, \mathcal{H}_K)]^T)||_2,
\end{equation}
or the weighted least deviation problem.
\begin{equation}\label{eq:lad}
\widehat{\theta} = \arg \min_{\theta} ||W^{\frac{1}{2}}(F\theta - [\widehat{q}(\mathcal{H}_1, \mathcal{H}_1), \ldots, \widehat{q}(\mathcal{H}_K, \mathcal{H}_K)]^T)||_1.
\end{equation}
The matrix $W$ is a diagonal weight matrix. The weighted least squares problem has a closed-form solution. The weighted least deviation problem can be solved numerically. Algorithm~\ref{alg} summarizes the approach to estimate $\theta$ from a finite-duration trajectory $\Sigma$. 

\begin{algorithm}[ht]
\SetAlgoLined
\KwResult{$\widehat{\theta}$}
Given continuous-time trajectory $\Sigma = \{x(t_0), x(t_1), x(t_2), \ldots, x(t_{M-1})$, find the holding class for each $X(t_i)$: $\mathcal{H}= \{ \mathcal{H}_1, \mathcal{H}_2, \ldots \mathcal{H}_m\}, m \le K$ to form the matrix $F$, which would be $m \times b$.

\For{$\mathcal{H}_i \in \mathcal{H}$}{
Estimate $\widehat{q}(\mathcal{H}_i, \mathcal{H}_i)$ from the given observations using equation~\eqref{eq:MLEdiagonaltransitionrate2}. Estimate the corresponding sample variance $var(\widehat{q}(\mathcal{H}_i, \mathcal{H}_i))$ to determine the diagonal entries of $W$. 


%
}
Estimate $\widehat{\theta}$ by solving the weighted least squares~\eqref{eq:wls} or weighted least deviation~\eqref{eq:lad}.  

\caption{Estimate $\theta$ from continuous-time trajectory $\Sigma$}\label{alg}
\end{algorithm}


\section{Contact Process}\label{sec:contact}

Interacting particle systems (IPS) models random interactions amongst $N$ particles (this paper assumes that $N$ is finite). Without interactions, the system would consist of $N$ independent continuous-time Markov processes. With interactions, the dynamic of 
dynamics of the particles becomes coupled and the evolution of the individual particles loses their Markovian property. 

Typically, the structure of interaction is characterized by an unweighted, undirected graph $G(V,E)$, where $|V| = N$. We will call this the \emph{contact network}. The most well-known IPS model is the contact process, which models infection and healing of the particles (i.e., nodes) according to the network-based susceptible-infected-susceptible (SIS) epidemics framework. IPS models such as the reversible contact process and the dynamic bond percolation processes have also been proposed and studied \cite{PhysRevE.86.016116, JZhang}. 

A Markov state $i \in \mathcal{D} = \{1, 2, \ldots D\}$ is a vector whose components describes the state of each node; we will refer to the state as a \emph{configuration} when we wish to emphasize the graphical nature of the system. The contact process assumes that each node can be in one of two states, $\{0,1\}$, representing healthy or infected state respectively. At time $t$, the network configuration is
\[
\mathbf{x}(t) = [x_1(t), x_2(t), \ldots x_N(t)]^T, \text{ where } x_i(t) = \{0,1\}.
\]
We see then that $D = 2^N$. The size of the state space is exponential in the number of nodes in the contact network. 

The contact process assumes that multiple nodes can not change states simultaneously. There are two types of transitions: 1) healing of infected agents and 2) infection of susceptible agents.

\begin{enumerate}
\item
Consider a configuration
\[ 
\mathbf{x} = [x_1,x_2, \ldots, x_j = 1, x_{j+1}, \ldots x_N]^T.
\] 
Let $T^-_j\mathbf{x}$ be the configuration where the $j$th node heals: 
\[
T^-_j\mathbf{x} =  [x_1,x_2, \ldots, x_j = 0, x_{j+1}, \ldots x_N]^T.
\] 
The transition rate from state $\mathbf{x}$ to $T^-_j\mathbf{x}$ is
\begin{equation}\label{eq:contactheal}
q(\mathbf{x}, T^-_j\mathbf{x}) = \mu,
\end{equation}
where $\mu \ge 0$ is the \emph{healing rate} (num. of healing events/unit time). If infected nodes can not heal, then $\mu = 0$.

\item Consider a configuration 
\[
\mathbf{x} = [x_1,x_2, \ldots, x_{k-1}, x_k = 0, \ldots x_N]^T.
\]
Let $T^+_k\mathbf{x}$ be the configuration where the $k$th node becomes infected:
\[
T^+_k\mathbf{x} = [x_1,x_2, \ldots, x_{k-1}, x_k = 1, \ldots x_N]^T.
\]
The transition rate from state $\mathbf{x}$ to $T^+_k\mathbf{x}$ is
\begin{equation}\label{eq:contactinfect}
q(\mathbf{x}, T^+_k\mathbf{x}) = \beta + \delta m_k, 
\end{equation}
where $m_k$ is the number of infected neighbors of node $k$ in configuration $\mathbf{x}$. Let $A = [A_{ik}]$ be the adjacency matrix of the contact network $G(V,E)$, then
\begin{equation}\label{eq:numinfneighbors}
m_k = \sum_{i =1}^N x_iA_{ik}.
\end{equation}

The parameter $\beta \ge 0$ is the \emph{exogeneous infection rate} (num. of infection from outside the network/unit time) and $\delta \ge 0$ is the \emph{endogeneous infection rate} (num. of contagion from infected neighbors/unit time). When $\delta> 0$, the infection rate of the contact process is linearly dependent on the number of infected neighbors. If susceptible nodes can only be infected by neighboring infected nodes (i.e., there can be no infection from outside the network), then $\beta= 0$.
\end{enumerate}

While the transition rate matrix, $Q$, of the contact process is $2^N \times 2^N$, we see can see from~\eqref{eq:contactheal} and~\eqref{eq:contactinfect} that the transition rates are linear functions of only three parameters: healing rate, $\mu$, exogeneous infection rate, $\beta$, and endogeneous infection rate, $\delta$.

\subsection{Upperbound on the Number of Holding Classes}

\begin{lemma}\label{prop:1}
Let $\mathcal{S}(\mathbf{x}) \subset V$ denote the set of susceptible nodes in a configuration $\mathbf{x}$. For the contact process, two different configurations $\mathbf{x}$ and $\mathbf{x}'$ belongs to the same holding class, $\mathcal{H}$, if and only if
\begin{equation}\label{eq:contactcond1}
| \mathcal{S}(\mathbf{x}) | = |\mathcal{S}(\mathbf{x}') |, \text{ and}
\end{equation}
\begin{equation}\label{eq:contactcond2}
\sum_{k \in \mathcal{S}(\mathbf{x})} m_k = \sum_{k \in \mathcal{S}(\mathbf{x}')} m_k,
\end{equation}
where $m_k$ is the total number of infected neighbors of node $k$. For the contact process, the holding classes are determined by the total number of susceptible nodes and the sum of infected neighbors of all the susceptible nodes.  
\end{lemma}

\begin{proof}
For two Markov states, corresponding to configurations $\mathbf{x}$ and $\mathbf{x}'$, the diagonal entries of the transition rate matrix $Q$ are 
\begin{align*}
&q(\mathbf{x}, \mathbf{x}) = -\sum_{\tilde{\mathbf{x}} \in \mathcal{D}} q(\mathbf{x}, \tilde{\mathbf{x}})\\
&= -\left((N- | \mathcal{S}(\mathbf{x}) |)\mu + (| \mathcal{S}(\mathbf{x}) |) \beta +   \left(\sum_{k \in \mathcal{S}(\mathbf{x})} m_k\right)\delta \right)
\end{align*}
and
\begin{align*}
&q(\mathbf{x}', \mathbf{x}') = -\sum_{\tilde{\mathbf{x}} \in \mathcal{D}} q(\mathbf{x}', \tilde{\mathbf{x}})\\
&= -\left((N- | \mathcal{S}(\mathbf{x}') |)\mu + (| \mathcal{S}(\mathbf{x}') |) \beta +   \left(\sum_{k \in \mathcal{S}(\mathbf{x}')} m_k\right)\delta \right).
\end{align*}

When equations~\eqref{eq:contactcond1} and~\eqref{eq:contactcond2} are true, then $q(\mathbf{x}, \mathbf{x}) = q(\mathbf{x}', \mathbf{x}')$. By definition~\ref{def:holdingclass}, Markov states $\mathbf{x}
$ and $\mathbf{x}'$ belongs to the same holding class, $\mathcal{H}$.
\end{proof}

Using Preposition~\ref{prop:1}, we can partition the $2^N$-sized state space of the contact process into holding classes $\{\mathcal{H}_1 \cup \mathcal{H}_2 \cup \ldots  \cup\mathcal{H}_K\}$. The total number of holding classes, $K$, depends on \emph{both} the infection/healing rates and the structure of the contact network $G(V,E)$. It is not easy to find $K$ without enumerating over all $2^N$ possible configurations. However, we can derive an upperbound on $K$ to see that it is much smaller than $2^N$.


\begin{theorem}\label{thm:1}
For a contact process with contact network $G(V,E), |V| = N$ and dynamics parameter $\theta$, the number of holding class 
\[
K \le \frac{(N+1)(N^2 - N +6)}{6}.
\]
\end{theorem}

\begin{proof}
Consider Lemma~\ref{prop:1}. If the holding class is only determined by~\eqref{eq:contactcond1}, then there would be $N+1$ different holding classes corresponding to $|\mathcal{S}(\mathbf{x})| = 0, 1, \ldots N$. It is intuitive that the number of holding classes should be much larger than $O(N)$. Therefore, the number of holding classes, $K$, is determined by how many possible unique values $\sum_{k \in \mathcal{S}(\mathbf{x})} m_k$ may take for all $\mathbf{x} \in \mathcal{D}$. 

Assuming that there are $|\mathcal{S}(\mathbf{x})|$ susceptible nodes in a configuration $\mathbf{x}$, this means that there are $N- |\mathcal{S}(\mathbf{x})|$ total number of infected nodes. Without any additional knowledge of the structure of $G(V,E)$, we can conclude that that for any node $k \in \mathcal{S}(\mathbf{x})$, $m_k$ can be take on \emph{at most} $N- |\mathcal{S}(\mathbf{x})| +1 $ different values (i.e.,$\{0, 1, 2, \ldots N- |\mathcal{S}(\mathbf{x})|\}$). Then,




Let $|\mathcal{S}(\mathbf{x})| = s$, summing over all possible values of $s$ from $0$ to $N$,
\begin{align}\label{eq:contactsum}
\sum_{s=0}^N &s(N-s) +1 = N\sum_{s=0}^N s -  \sum_{s=0}^N s^2 +  \sum_{s=0}^N 1\\
&\medmath{= N\left(\frac{N(N+1)}{2}\right) - \frac{N(N+1)(2N+1)}{6} +(N+1)}\\
&\medmath{=\frac{(N+1)(N^2 - N +6)}{6}}.
\end{align}
Therefore
\[
 K \le \frac{(N+1)(N^2 - N +6)}{6}.
\]

\end{proof}

Theorem~\ref{thm:1} shows that the number of holding classes can not be more than cubic in the number of nodes, $N$. This makes $K$ much smaller than $D = 2^N$. 

 \begin{corollary}
For a contact process with contact network $G(V,E), |V| = N$ and dynamics parameter $\theta$, if we know that $G(V,E)$ is a degree-bounded graph such that all the nodal degrees are less than or equal to $d_{\max}$, then the number of holding class 
\begin{equation}\label{eq:upperbound}
K \le \sum_{s = 0}^N s(\min(N-s, d_{\max}))+1 \le \frac{(N+1)(N^2 - N +6)}{6}.
\end{equation}

 \end{corollary}

\begin{proof}
From the proof of Theorem~\ref{thm:1}, for a given number of susceptible nodes, $s$, we assumed that the number of infected neighbors of a susceptible node can range from $0, 1, \ldots, N-s$. When the contact network has bounded degree, we know that the number of infected neighbors of a susceptible node can range from $0,1, \ldots, \min(N-s, d_{\max})$. Therefore, \eqref{eq:contactsum} becomes
\begin{align*}
\sum_{s=0}^N s(\min(N-s, d_{\max})) + 1.
\end{align*}
\end{proof}

 \begin{lemma}
For a contact process with contact network $G(V,E), |V| = N$ and dynamics parameter $\theta$, if we know that $G(V,E)$ is a complete graph, then the number of holding class 
\[
K = N+1.
\]
 \end{lemma}
\begin{proof}
Let $\mathbf{x}$ and $\mathbf{x}'$ denote any two configurations such that $|\mathcal{S}(\mathbf{x})| = |\mathcal{S}(\mathbf{x}')|$. Then we know that 
\[
\sum_{s \in \mathcal{S}(\mathbf{x})} m_s = \sum_{s \in \mathcal{S}(\mathbf{x'})} m_s = s(N-s), \forall \mathbf{x}, \mathbf{x}'.
\]
Since $|\mathcal{S}(\mathbf{x})|$ can be $0,1,\dots N$, the number of holding classes is $N+1$.

\end{proof}

\section{Reversible Contact Process}\label{sec:rev}

In the contact process, we see that the infection rate of a susceptible node is linearly dependent on the number infected neighbors. In~\cite{Zhang2014}, a process similar to the contact process (called scaled SIS process) was analyzed. In this model, the infection rate is exponentially dependent on the number of infected neighbors. This modification made the underlying continuous-time Markov process a \emph{reversible} process: a stochastic process that is statistically the same forward and backward in time. Additionally, the equilibrium distribution of the reversible contact process can be derived in closed-form. 

The transition rates of the reversible contact process are
\begin{enumerate}
\item 
\begin{equation}\label{eq:scaledheal}
q(\mathbf{x}, T^-_j\mathbf{x}) = \mu,
\end{equation}
where $\mu > 0$ is the \emph{healing rate} (\# of healing events/unit time). The reversible contact process can not have $\mu = 0$, or the reversibility property would be lost. 

\item 
\begin{equation}\label{eq:scaledinfect}
q(\mathbf{x}, T^+_k\mathbf{x}) = \beta(\delta)^{m_k}, 
\end{equation}
where $m_k$ is the number of infected neighbors of node $k$ in configuration $\mathbf{x}$. Like in the contact process, we can think of $\beta >0$ as the exogeneous infection rate (\# of infected events/unit). Unlike the contact process, $\delta$ is not a rate but a unitless factor. We see that the number of infected neighbors, $m_k$, induces a scaling of the infection rate from $\beta$. This scaling has increases the infection rate with increasing $m_k$ if $\delta > 1$ and decreases infection rate if $\delta < 1$. The reversible contact process can not have $\beta =0$ or the reversibility property would be lost. 

\end{enumerate}

The transition rate matrix, $Q$, of the reversible contact process is $2^N \times 2^N$. We can see from~\eqref{eq:scaledheal} and~\eqref{eq:scaledinfect} that unlike the contact process, the transition rates are not linear functions of $\mu, \beta, \delta$. Instead, the transition rates can be written as functions of $\theta = [\mu, \beta, \beta\delta, \beta\delta^2, \ldots \beta\delta^{\text{dmax}}]^T$, where $dmax$ is the maximum degree of $G(V,E)$. A simple additional step is needed after Algorithm~\ref{alg} to compute $\delta$ from $\beta\delta, \ldots \beta\delta^{\text{dmax}}$, but this may induce additional errors.

\subsection{Upperbound on the Number of Holding Classes}



\begin{lemma}\label{lemmascaled}
Let $\mathcal{S}(\mathbf{x}) \subset V$ denote the set of susceptible nodes in a configuration $\mathbf{x}$. For the reversible contact process, two different configurations $\mathbf{x}$ and $\mathbf{x}'$ belongs to the same holding class if and only if
\begin{equation}\label{eq:contactcond3}
| \mathcal{S}(\mathbf{x}) | = |\mathcal{S}(\mathbf{x}') |, \text{ and}
\end{equation}
\begin{equation}\label{eq:contactcond4}
\{m_k: k \in S(\mathbf{x})\} = \{m_k: k \in S(\mathbf{x}')\},
\end{equation}
where $m_k$ is the number of infected neighbors of node $k$. The set $\{m_k: k \in S(\mathbf{x})\}$ is the set of the number of infected neighbors of each susceptible node in configuration $\mathbf{x}$. Equality in~\eqref{eq:contactcond4} is set equality; the ordering of the values do not matter.
\end{lemma}
\begin{proof}
For two Markov states, corresponding to configurations $\mathbf{x}$ and $\mathbf{x}'$, the diagonal entries of the transition rate matrix $Q$ are 
\begin{align*}
&q(\mathbf{x}, \mathbf{x}) = -\sum_{\tilde{\mathbf{x}} \in \mathcal{D}} q(\mathbf{x}, \tilde{\mathbf{x}})\\
&\medmath{= -((N- | \mathcal{S}(\mathbf{x}) |)\mu +  |{ s \in S(\mathbf{x}): m_s = 0 }|(\beta)}+\\
&\medmath{ |\{s \in S(\mathbf{x}): m_s = 1\}|(\beta\delta) \ldots + |\{s \in S(\mathbf{x}): m_s = d_{\max}\}|(\beta\delta^{d_{\max}}))}, 
\end{align*}
and
\begin{align*}
&q(\mathbf{x}', \mathbf{x}') = -\sum_{\tilde{\mathbf{x}} \in \mathcal{D}} q(\mathbf{x}', \tilde{\mathbf{x}})\\
&= \medmath{-((N- | \mathcal{S}(\mathbf{x}') |)\mu +  (|{ s \in S(\mathbf{x}'): m_s = 0 }|)\beta}+\\
&\medmath{(|\{s \in S(\mathbf{x}'): m_s = 1\}|)\beta\delta \ldots + (|\{s \in S(\mathbf{x}'): m_s = d_{\max}\}|)\beta\delta^{d_{\max}}}.
\end{align*}

When equations~\eqref{eq:contactcond3} and \eqref{eq:contactcond4} are true, then $q(\mathbf{x}, \mathbf{x}) = q(\mathbf{x}', \mathbf{x}')$. By Definition~\ref{def:holdingclass}, Markov state $\mathbf{x}
$ and $\mathbf{x}'$ belongs to the same holding class, $\mathcal{H}$.
\end{proof}

In contrast the contact process, the holding classes of the reversible process is determined by the number of infected neighbors of each susceptible node instead of simply by the total number of infected neighbors. Consequently, the number of holding classes in the reversible process is larger than the number of holding classes in the contact process. Reversible contact processes also have a larger number of values to estimate $\theta = [\mu, \beta, \beta\delta, \beta\delta^2, \ldots \beta\delta^{\text{dmax}}]^T$. It may be useful then, to approximate the reversible process by the contact process to efficiently estimate the healing and infection rates. Reference~\ref{Zhang2017} gives some insight as to when the two processes are equivalent.

\begin{theorem}\label{thm:revtran}
For a reversible contact process with interaction network $G(V,E), |V| = N$ and dynamics parameter $\theta$, the number of holding class 
\begin{equation}\label{eq:revbound}
K \le 2^N.
\end{equation}
\end{theorem}

\begin{proof}
If the holding class is only determined by~\eqref{eq:contactcond3}, then there would be $N+1$ different holding classes corresponding to $|\mathcal{S}(\mathbf{x})| = 0, 1, \ldots N$. It is intuitive that the number of holding classes should be much larger than $O(N)$. Therefore, the number of holding classes, $K$, is determined by the number of unique sets $\{m_k: k \in S(\mathbf{x})\}$.

Assuming that there are $|\mathcal{S}(\mathbf{x})|$ susceptible nodes in a configuration $\mathbf{x}$, this means that there are $N- |\mathcal{S}(\mathbf{x})|$ total number of infected nodes. Without any additional knowledge of the structure of $G(V,E)$, we can conclude that that $m_k$ can range in value from $\{0,1, \ldots N- |\mathcal{S}(\mathbf{x})|\}$ for any node $k \in \mathcal{S}(\mathbf{x})$. The number of possible unique set $\{m_k: k \in S(\mathbf{x})\}$ is a combination with replacement problem where we want to choose $N-|\mathcal{S}(\mathbf{x})|+1$ values for $|\mathcal{S}(\mathbf{x})|$ entires. 

Let $|\mathcal{S}(\mathbf{x})| =s$. The number of possibilities is
\[
\frac{(N-s+1+s - 1)!}{s!(N-s+1-1)!} = \frac{N!}{s!(N-s)!}.
\]
Summing over all possible number of susceptible nodes results in
\begin{equation}\label{eq:revsum}
\sum_{s=0}^N \frac{N!}{s!(N-s)!}.
\end{equation}
By the binomial theorem, we know that~\eqref{eq:revsum} is equal to $2^N$.
\end{proof}

 \begin{corollary}\label{coro:revtran}
For a reversible contact process with interaction network $G(V,E), |V| = N$ and dynamics parameter $\theta$, if we know that $G(V,E)$ is a degree-bounded graph such that all the nodal degrees are less than or equal to $d_{\max}$, then the number of holding class 
\[
K \le \sum_{s=0}^{N - d_{\max}}  \frac{(d_{\max} +s)!}{s!(d_{\max})!}  + \sum_{s = N -d_{\max}+1}^{N} \frac{N!}{s!(N-s)!}.
\]
 \end{corollary}

\begin{proof}
From the proof of Theorem~\ref{thm:revtran}, for a given number of susceptible nodes, $s$, we assumed that the number of infected neighbors of a susceptible node can range from $0, 1, \ldots, N-s$. When the interaction network has bounded degree, we know that the number of infected neighbors of a susceptible node can range from $0,1, \ldots, \min(N-s, d_{\max})$. Therefore, the sum~\eqref{eq:revsum} becomes
\begin{align*}
&\sum_{s=0}^N \frac{N!}{s!(N-s)!}\\ 
&= \medmath{\sum_{s=0}^{N - d_{\max}}  \frac{(d_{\max} +s)!}{s!(d_{\max})!}      + \sum_{N -d_{\max}+1}^{N} \frac{(N-s+1+s - 1)!}{s!(N-s+1-1)!}}\\
& = \sum_{s=0}^{N - d_{\max}}  \frac{(d_{\max} +s)!}{s!(d_{\max})!}      + \sum_{s = N -d_{\max}+1}^{N} \frac{N!}{s!(N-s)!}
\end{align*}

\end{proof}

\section{Number of Holding Classes and Contact Network Structure}\label{sec:numholdingexp}
 \begin{figure*}[ht]
\begin{subfigure}{.5\textwidth}
  \centering
  \includegraphics[width=0.7\linewidth]{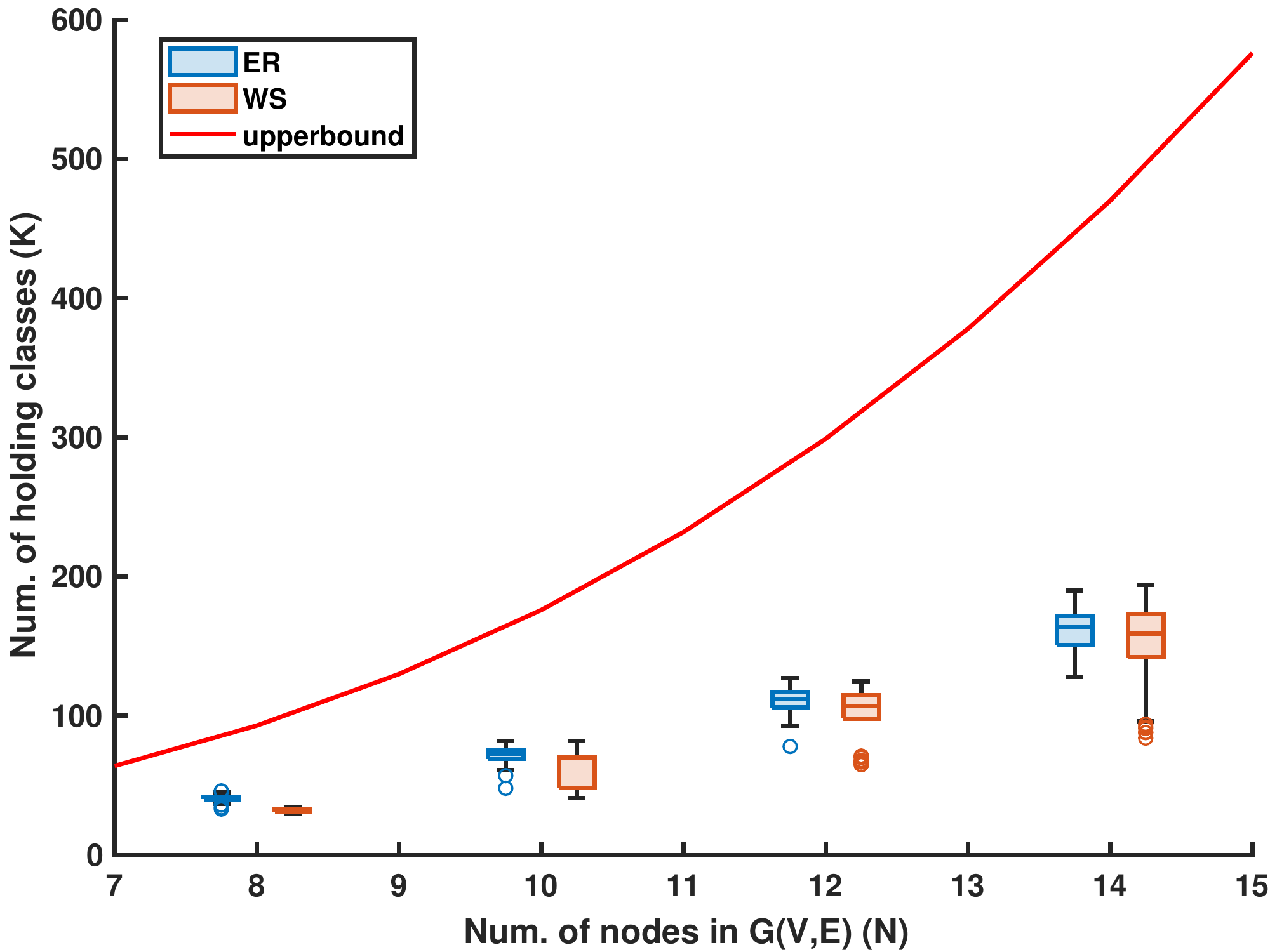}  
  \caption{Line indicates the theoretical upperbound~\eqref{eq:upperbound}}
  \label{fig:contactclass}
\end{subfigure}
\begin{subfigure}{.5\textwidth}
  \centering
  \includegraphics[width=0.7\linewidth]{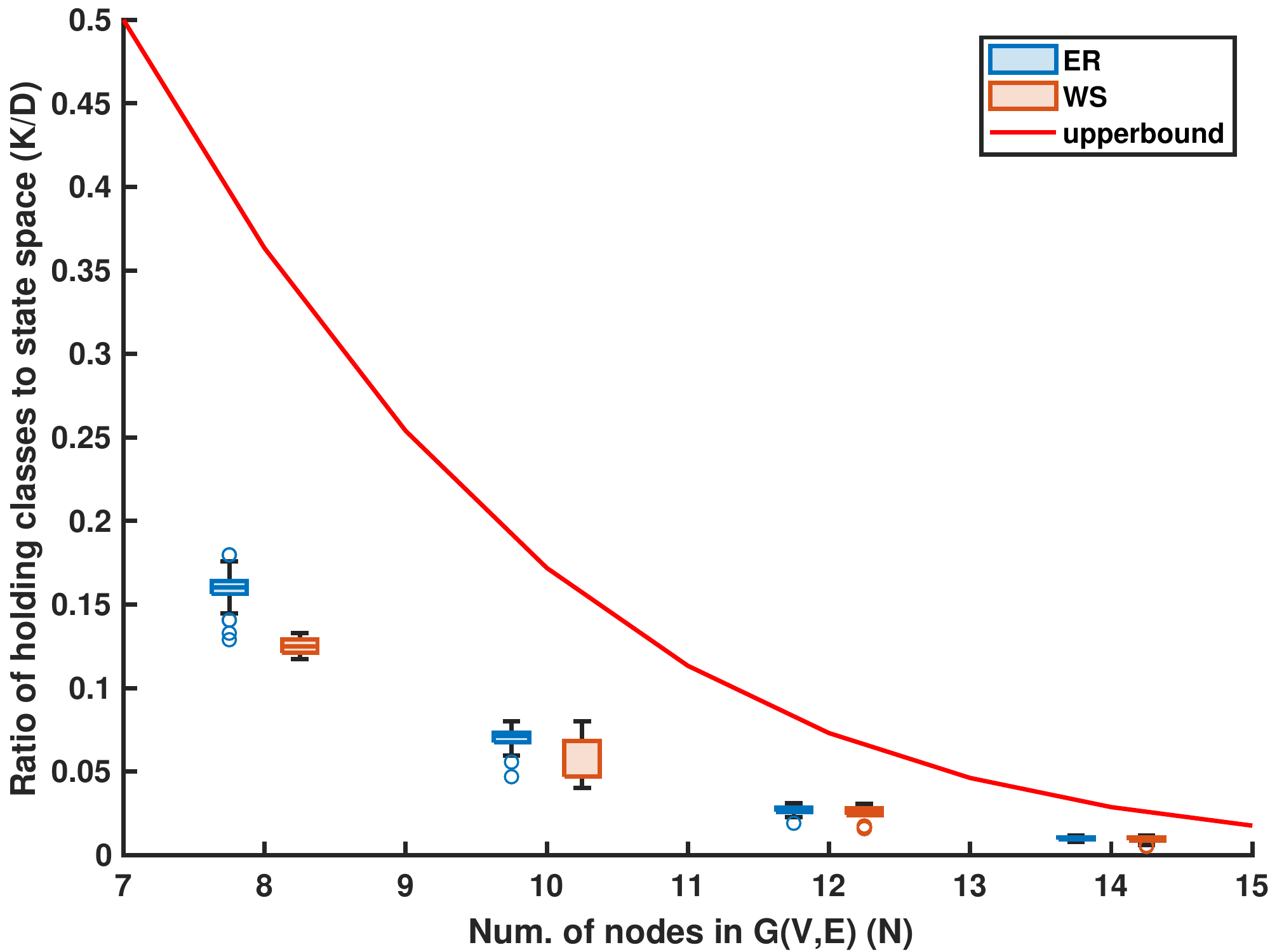}  
  \caption{Line indicates the theoretical upperbound~\eqref{eq:upperbound} divided by $D=2^N$}
  \label{fig:contactclassratio}
\end{subfigure}
\caption{Contact Process Num. of Holding Classes $(K)$ and Ratio to State Space Size $(K/D)$}
\end{figure*}

 \begin{figure*}[ht]
\begin{subfigure}{.5\textwidth}
  \centering
  \includegraphics[width=0.7\linewidth]{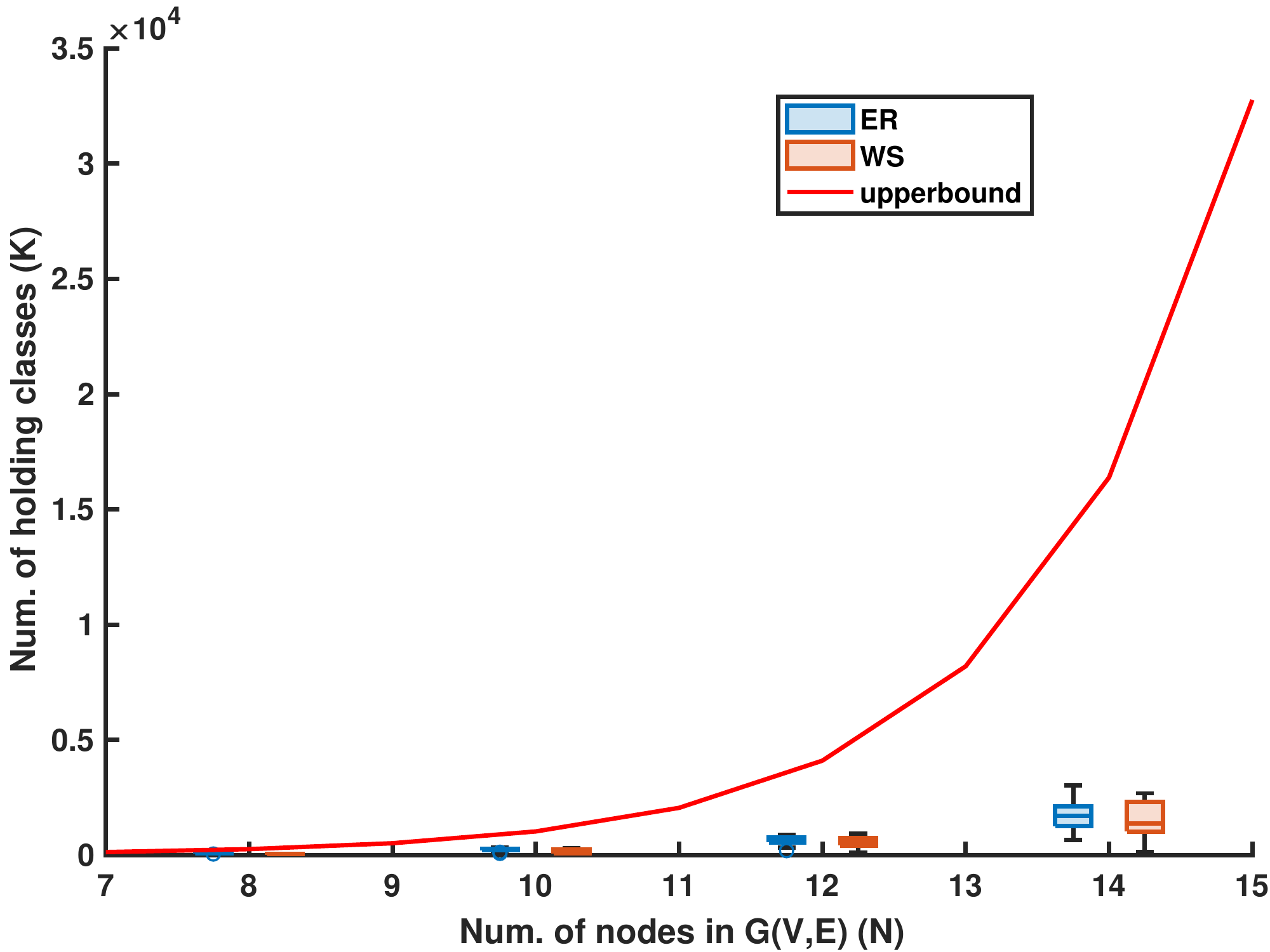}  
  \caption{Line indicates the theoretical upperbound~\eqref{eq:revbound}}
  \label{fig:revclass}
\end{subfigure}
\begin{subfigure}{.5\textwidth}
  \centering
  \includegraphics[width=0.7\linewidth]{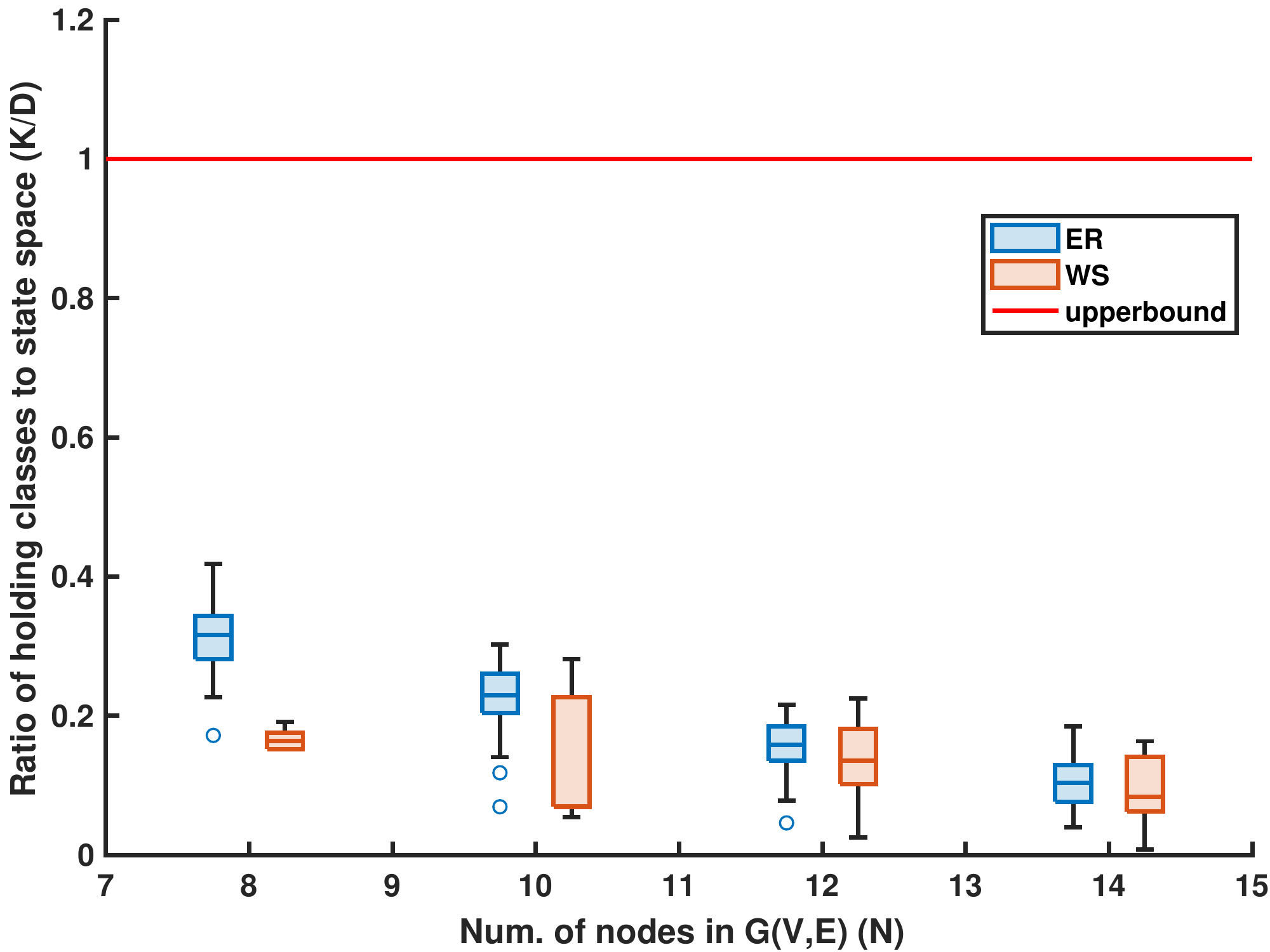}  
  \caption{Line indicates the theoretical upperbound~\eqref{eq:revbound} divided by $D= 2^N$}
  \label{fig:revclassratio}
\end{subfigure}
\caption{Reversible Contact Process Num. of Holding Classes $(K)$ and Ratio to State Space Size $(K/D)$}
\end{figure*}
 
The total number of holding classes, $K$, depends on the topology of the contact network, $G(V,E)$ and the formulation of the transition rates. For small sized networks, we can numerically compute the number of holding classes. First, we generated 50 different Erd\H{o}s-R\'{e}nyi (ER) graphs with the same number of nodes ($N$), but randomly chosen edge probability $p \sim U\left(\frac{\log N}{N}, 0.2\frac{\log N}{N}\right)$ assuming the contact process dynamics (\ref{eq:contactheal},~\ref{eq:contactinfect}) and the reversible contact process dynamics (\ref{eq:scaledheal},~\ref{eq:scaledinfect}). Then we generated 50 different Watts-Strogatz (WS) graphs of the same size $N$, but randomly chosen neighbor size, $nei \sim U\left(3, \frac{N}{2}\right)$ and rewiring probability $p \sim U(0.2, 0.8)$.

Figures~\ref{fig:contactclass} and~\ref{fig:revclass} show the total number of holding classes $K$ for different network size $N$ assuming the contact process and reversible contact process dynamics, respectively. The line shows the theoretical upperbound in~\eqref{eq:upperbound} and~\eqref{eq:revbound}. The actual number of holding classes is much smaller than the theoretical upperbound for both ER and WS graphs. Figures~\ref{fig:contactclassratio} and~\ref{fig:revclassratio} show the ratio of the number of holding class to the state space dimensionality. For the contact process, the ratio $\frac{K}{D}$ provably \emph{decreases} with increasing $N$ since the upperbound is polynomial instead of exponential in $N$. We see that this remains true for the reversible contact process experimentally.

\section{Parameter Estimation Experiments}\label{sec:exp}


We will study the performance of Algorithm~\ref{alg} for estimating the dynamics parameters for both the contact process and the reversible contact process using generated trajectories on synthetic and real-world graphs. We will look at two types of error: absolute error and relative error. The mean absolute error (MAE) is the absolute difference between the true healing rate, endeogenous infection rate, exogeneous infection rate ($\mu, \beta, \delta$) and the estimates ($\widehat{\mu}, \widehat{\beta}, \widehat{\delta}$): 
\[
\mu_\text{MAE} = \frac{1}{L} \sum_{i=1}^L |\widehat{\mu}_i- \mu_i |
\]

%

However, absolute error can be misleading when the true value is close to zero. Relative error compares the absolute error to the true and estimated values. We use symmetric mean absolute percentage error (SMAPE) \cite{flores1986pragmatic}:

\[
\mu_\text{SMAPE} = \frac{100\%}{L}\sum_{i=1}^{L} \frac{|\widehat{\mu}_i- \mu_i    |}{|\widehat{\mu}_i| + |\mu_i|},
\]
%
%
which is bounded between $0\%$ and $100\%$. A second reason that we used SMAPE is because an estimated value of $0$ when the true value is nonzero will have $100\%$ error. This is particularly important for learning parameters of graph dynamical systems because the parameters of interests are rates, which can be small but is generally nonzero.


\subsection{Parameter Estimation Performance for Contact Processes}

\subsubsection{\textbf{Synthetic Erd\H{o}s-R\'{e}nyi (ER) Random Graph}}

We generated fifty 100-node unweighted, undirected, connected Erd\H{o}s-R\'{e}nyi (ER) graphs. The state space of the contact process is $D = 2^{100}\approx 1.27e^{30}$. Per Theorem~\ref{thm:1}, we know that the maximum number of holding classes can be is 166,751. For each graph, a \emph{single} contact process trajectory of length $M$ was simulated, $\Sigma = \{\mathbf{x}(t_0), \ldots \mathbf{x}(t_{M-1})\}$. 


The initial configuration, $\mathbf{x}(t_0)$, is randomly generated; a probability $p_0$ is chosen from Uniform(0,1). The state of each node is independently assigned as a Bernoulli random variable with probability $p_0$. The infection and healing rates, $\mu, \beta, \delta$ are chosen independently from Uniform(0,3).

Given $\Sigma = \{\mathbf{x}(t_0), \ldots \mathbf{x}(t_{M-1})\}$, we used Algorithm~\ref{alg}. We used three different methods to find $\widehat{\theta}$: 1) weighted least squares (WLS), which solves~\eqref{eq:wls}; 2) non-negative (constrained) weighted least squares (NNLS), which solves~\eqref{eq:wls} but with the additional constraint that the estimated values need to be non-negative; 3) weighted least absolute deviation (LAD), which solves the weighed least absolute deviation problem~\eqref{eq:lad}.

\begin{figure*}[ht]
\begin{subfigure}{.5\textwidth}
  \centering
  \includegraphics[width=0.95\linewidth]{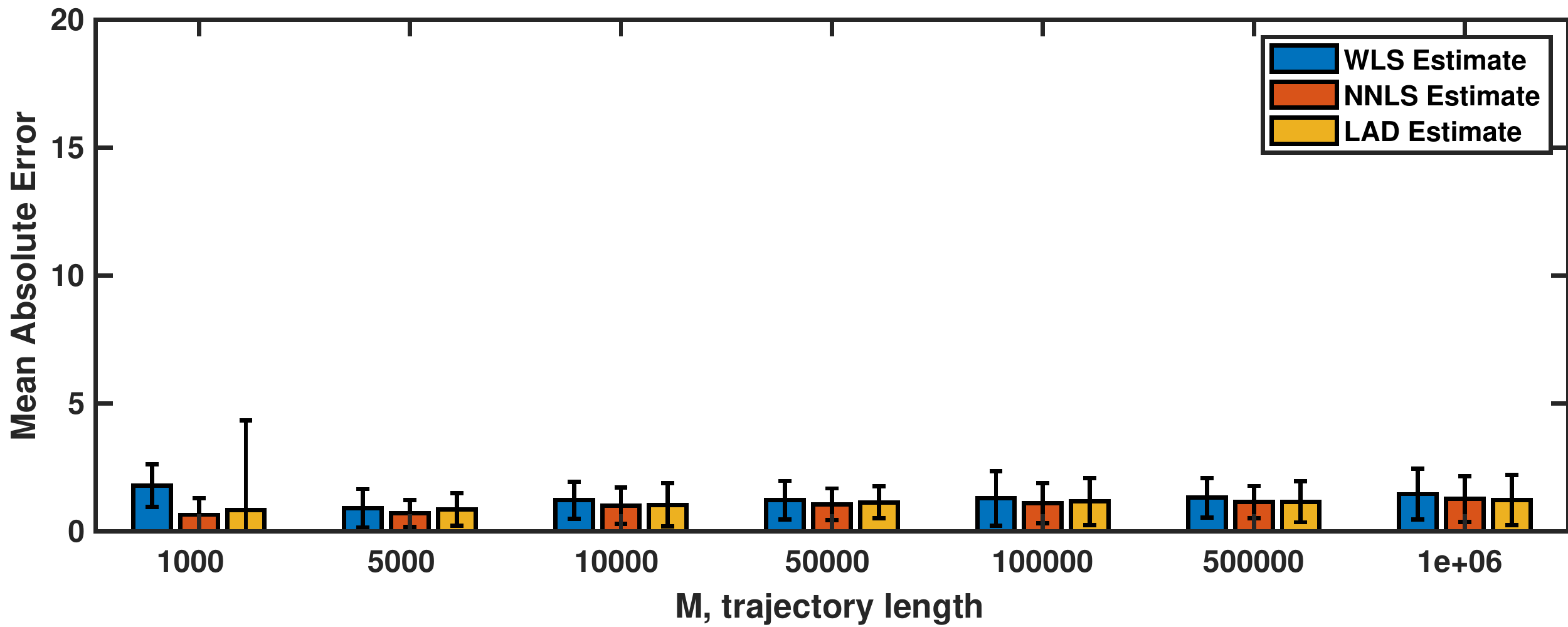}  
  \caption{MAE of Healing Rate $\mu$}
  \label{fig:contact_MAE_mu}
\end{subfigure}
\begin{subfigure}{.5\textwidth}
  \centering
  \includegraphics[width=0.95\linewidth]{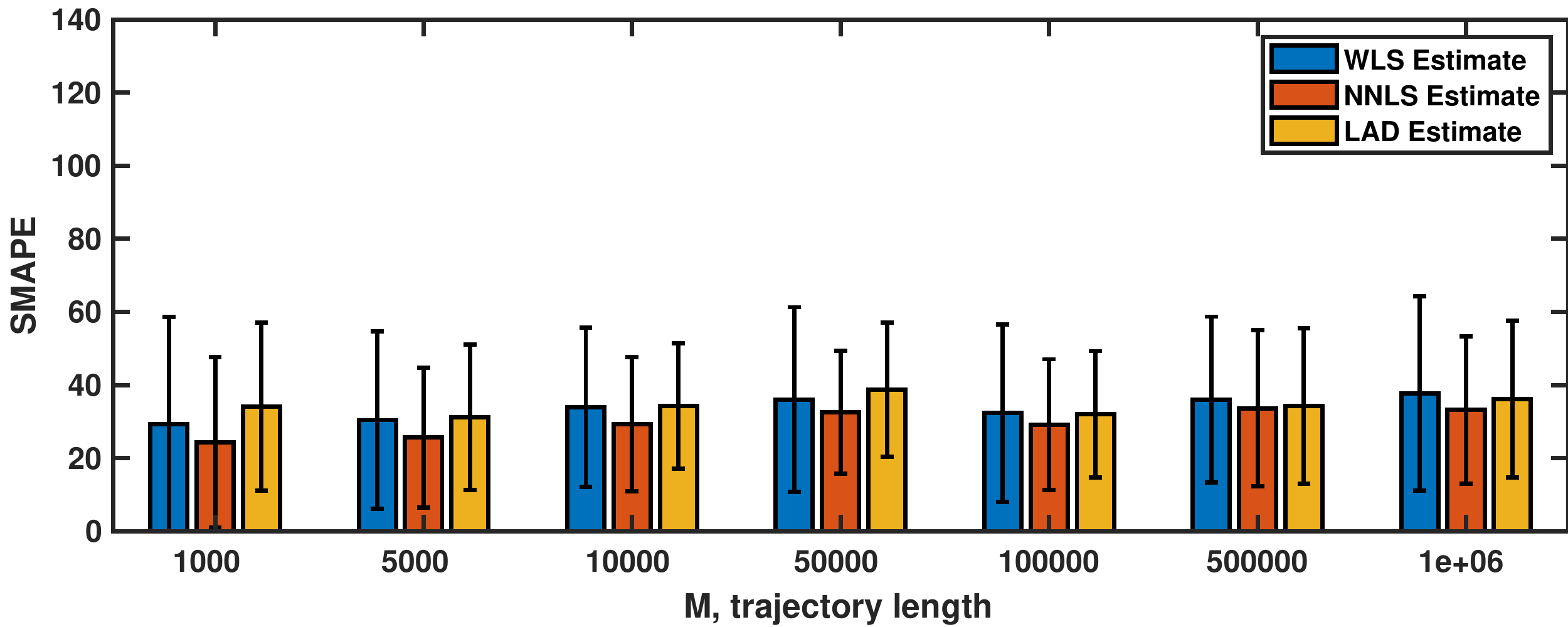}  
  \caption{SMAPE of Healing Rate $\mu$}
  \label{fig:contact_SMAPE_mu}
\end{subfigure}
\hfill

\begin{subfigure}{.5\textwidth}
  \centering
  \includegraphics[width=0.95\linewidth]{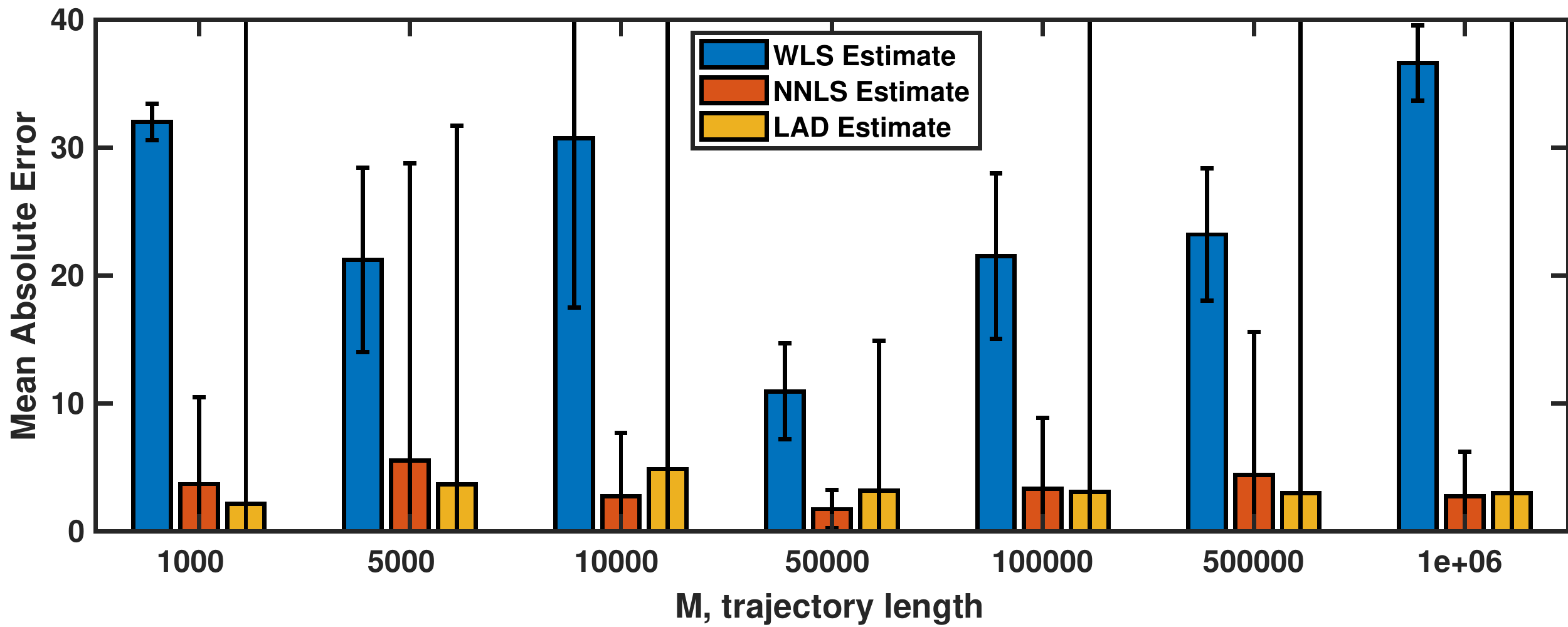}  
  \caption{MAE of Exogeneous Infection Rate $\beta$}
  \label{fig:contact_MAE_lambda}
\end{subfigure}
\begin{subfigure}{.5\textwidth}
  \centering
  \includegraphics[width=0.95\linewidth]{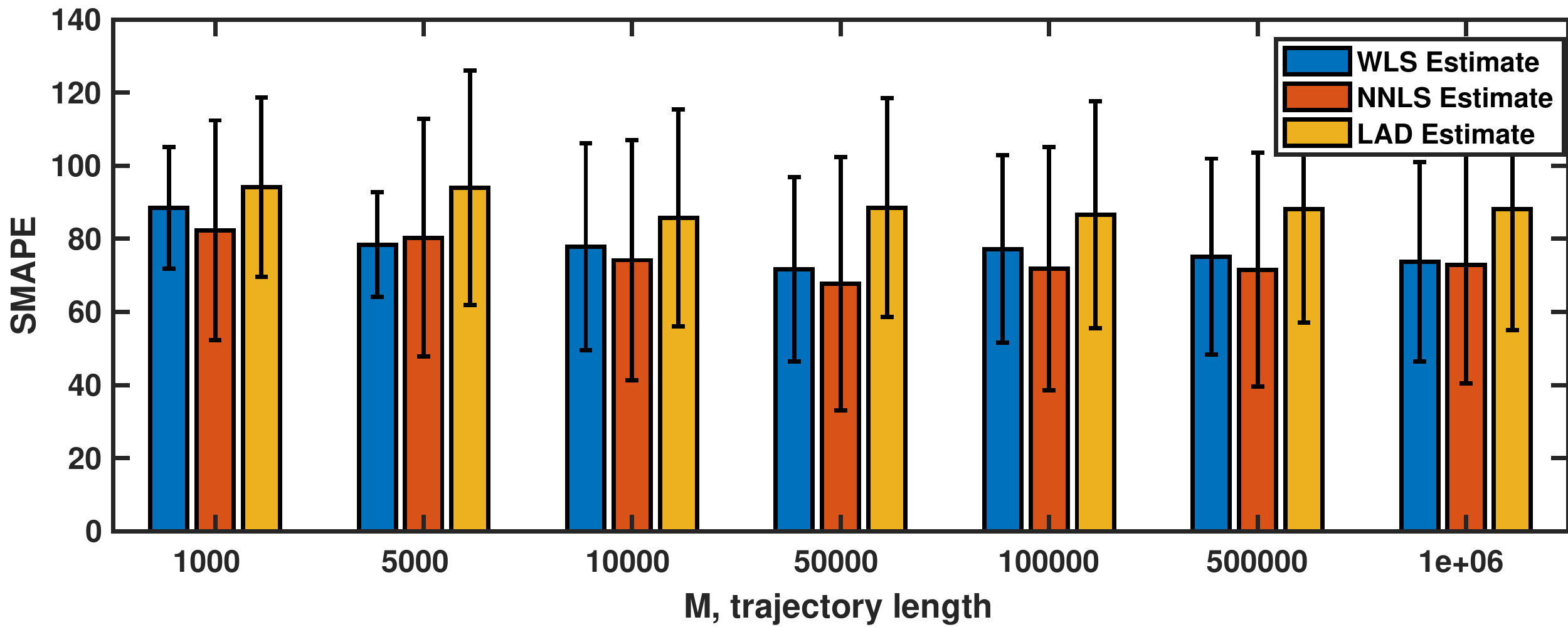}  
  \caption{SMAPE of Exogeneous Infection Rate $\beta$}
  \label{fig:contact_SMAPE_lambda}
\end{subfigure}
\hfill

\begin{subfigure}{.5\textwidth}
  \centering
  \includegraphics[width=0.95\linewidth]{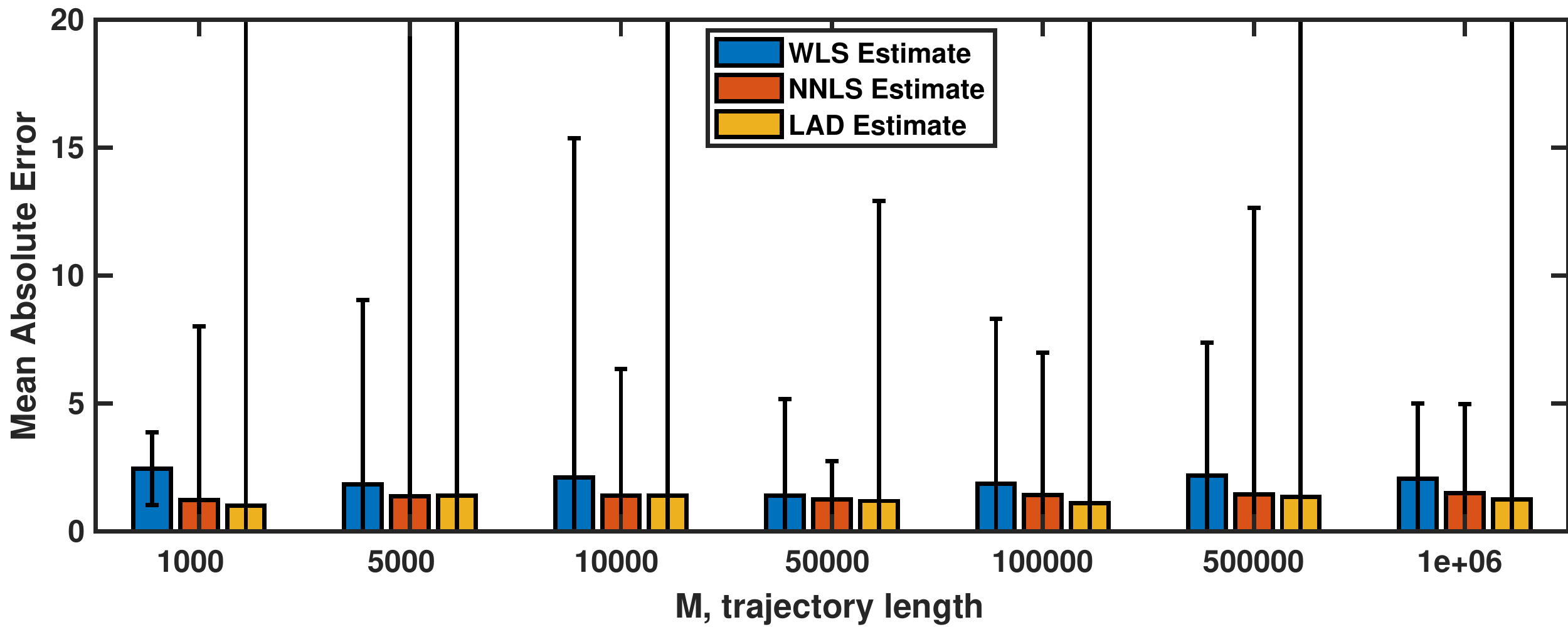}  
  \caption{MAE of Endogeneous Infection Rate $\delta$}
  \label{fig:contact_MAE_gamma}
\end{subfigure}
\begin{subfigure}{.5\textwidth}
  \centering
  \includegraphics[width=0.95\linewidth]{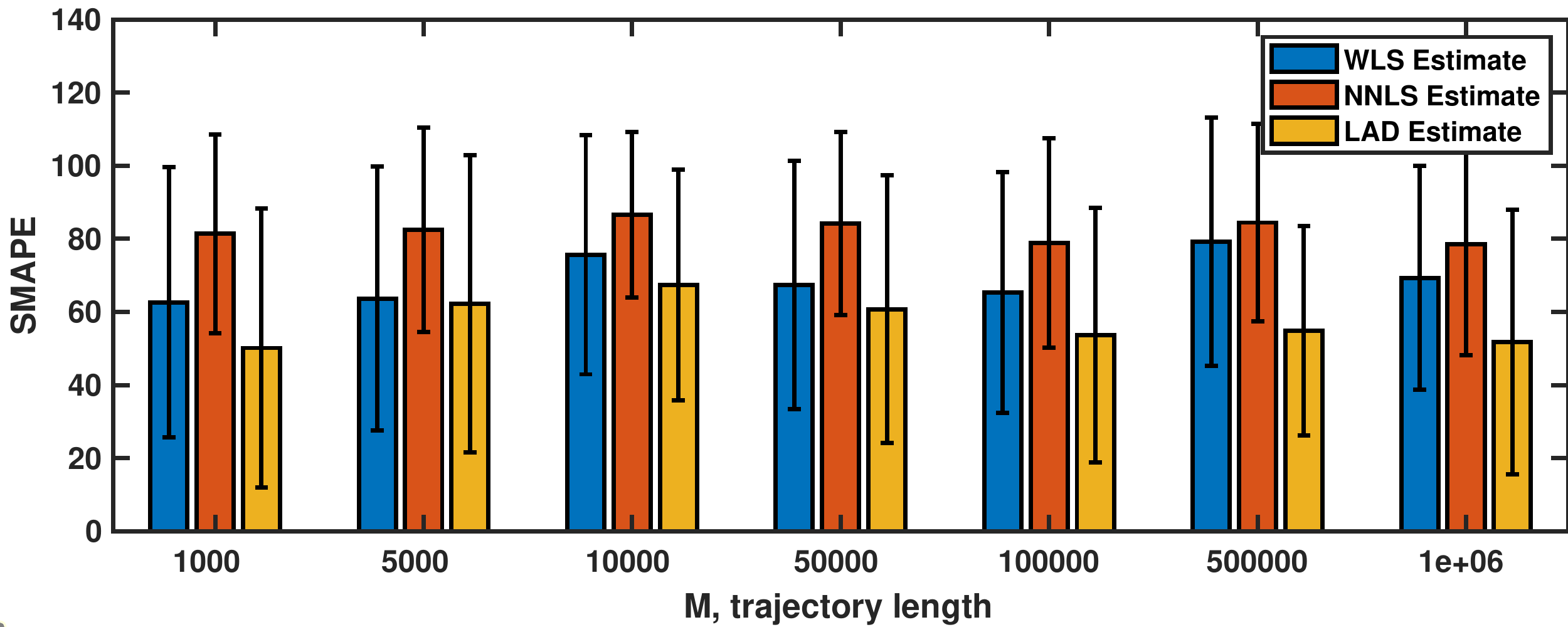}  
  \caption{SMAPE of Endogeneous Infection Rate $\delta$}
  \label{fig:contact_SMAPE_gamma}
\end{subfigure}
\caption{Contact Process: Error vs. Trajectory Length ($M$) for Erd\H{o}s-R\'{e}nyi Graphs}
\label{fig:contactER100result}
\end{figure*}

Figure~\ref{fig:contactER100result} shows the MAE and SMAPE error for the healing rate, endogeneous infection rate, and exogeneous infection rate. The black line indicates the standard deviation. The error does not decrease with increasing number of observations. Partly, this is because successive observations in the trajectory are \emph{not independent}. As a result the error is more influenced by the choice of the initial state and the topology of the $G(V,E)$. 

Looking at the absolute error, we can see that the standard deviation can sometimes be large be larger than the mean value. This means that the estimate is overdispersed. The weighted least squares (WLS) estimator has the largest error, but generally has the smallest standard deviation. If we consider relative error, all three estimator perform similarly. The reason that SMAPE of non-negative least squares (NNLS) and least absolute deviation (LAD) estimators increased compare to WLS is because both estimators may give results that are exactly zero. Therefore, they are penalized by SMAPE, which assigns a 100\% error to an estimate that is exactly zero. From this perspective, the performance of the WLS estimator is comparable to both NNLS and LAD.

The estimate of the healing rate always give the lowest error compared to estimates of infection rates. This is because there is only one type of healing events so there are comparatively more observations for estimating the healing rate. For the contact process, the error for the endogeneous infection rate, $\delta$ is smaller than the error for the exogeneous infection rate, $\beta$. One reason may be that direct observation of a node becoming infected via exogeneous infection (i.e., infection from outside the network) is rarely observed. Therefore, it can be difficult to disentangle $\beta$ from $\delta$. If we consider SMAPE, then NNLS estimator gives the lowest error for the endogeneous infection rate, while the LAD estimator gives the lowest error for the exogeneous infection rate.

\subsubsection{\textbf{118-Bus Power Grid Graph}}

\begin{figure}[h] 
\includegraphics[width=0.3\textwidth]{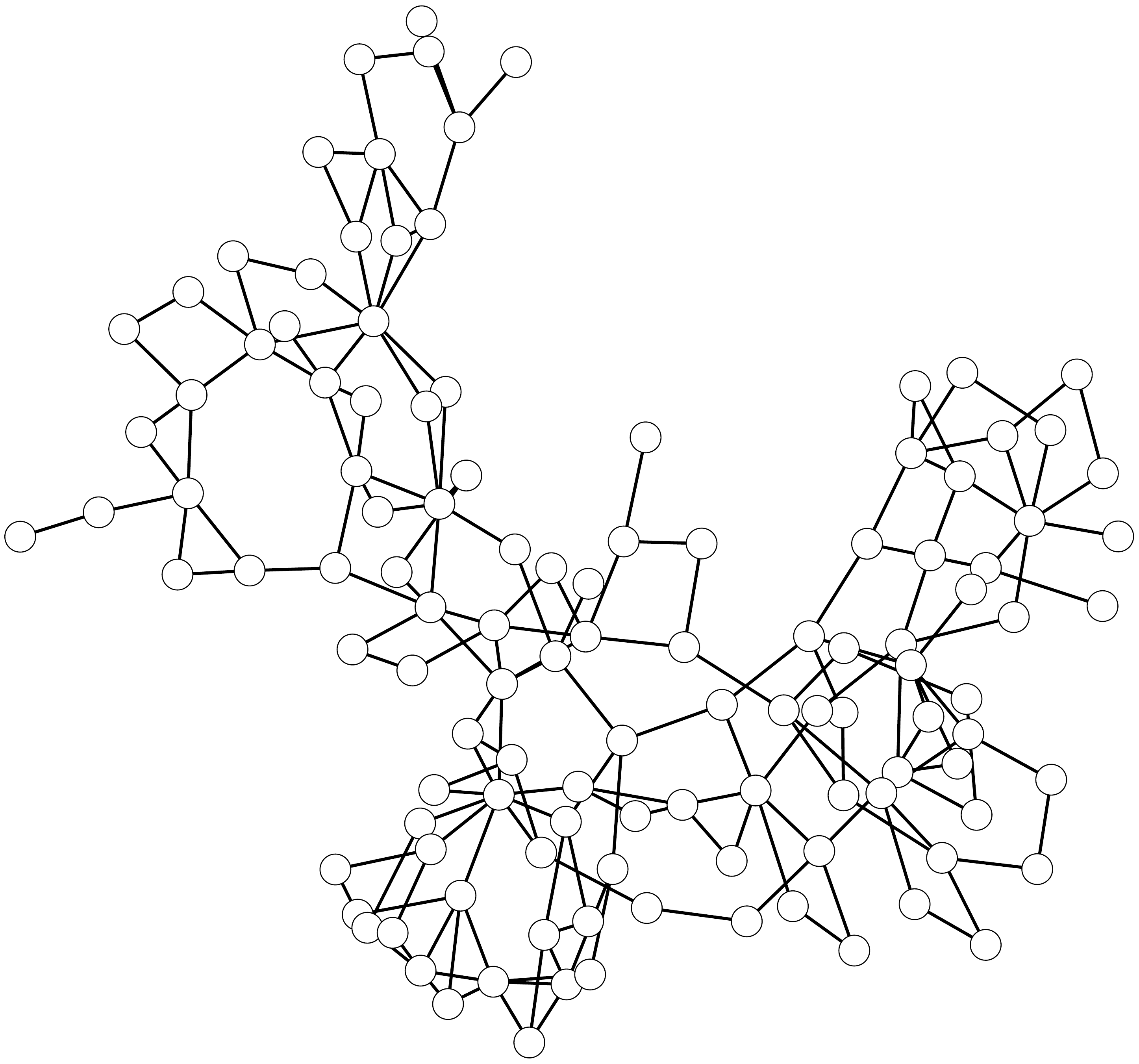}
\centering
 \caption{Network Visualization of 118-bus Power Flow Test System}
 \label{fig:118bus}
\end{figure}

In this section, we consider a real-world network, the IEEE 118-bus Power Flow Test System from \cite{christie1993power}. The nodes (N = 118) of the network correspond to the bus in the power grid, which is the location where a line or several lines connects at. The number of edges is 179. The maximum degree of the network is 9.  A node can be in the failed ($x_i =1$) or working ($x_i = 0)$ state. The graph representation of the system is shown in Figure~\ref{fig:118bus}.

For this graph, \emph{fifty} independent trajectories of length $M$, $\Sigma = \{\mathbf{x}(t_0), \ldots \mathbf{x}(t_{M-1})\}$ from the contact process was simulated. The initial configuration $\mathbf{x}(t_0)$ and dynamics parameters, $\mu, \beta, \delta$, are chosen randomly using the same procedure described previously. Figure~\ref{fig:contactpowerresult} shows the MAE and SMAPE error for $\mu, \beta, \delta$. Again we see that the length of the observed trajectory do not have much effect on the error. The error is smaller compared to experiments involving random graphs. This suggest that the topology of the network has a large impact on how well parameters can be learned from data.


\begin{figure*}[ht]
\begin{subfigure}{.5\textwidth}
  \centering
  \includegraphics[width=0.95\linewidth]{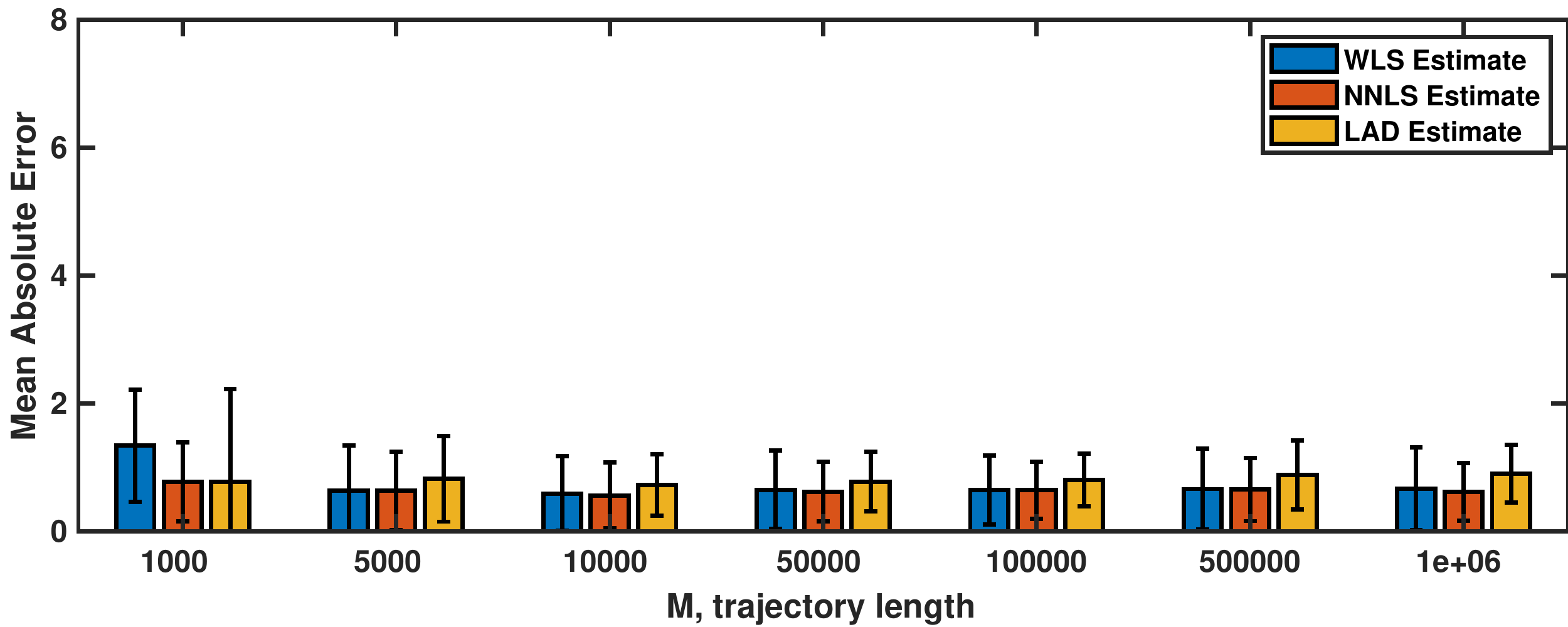}  
  \caption{MAE of Healing Rate $\mu$}
  \label{fig:118buscontact_MAE_mu}
\end{subfigure}
\begin{subfigure}{.5\textwidth}
  \centering
  \includegraphics[width=0.95\linewidth]{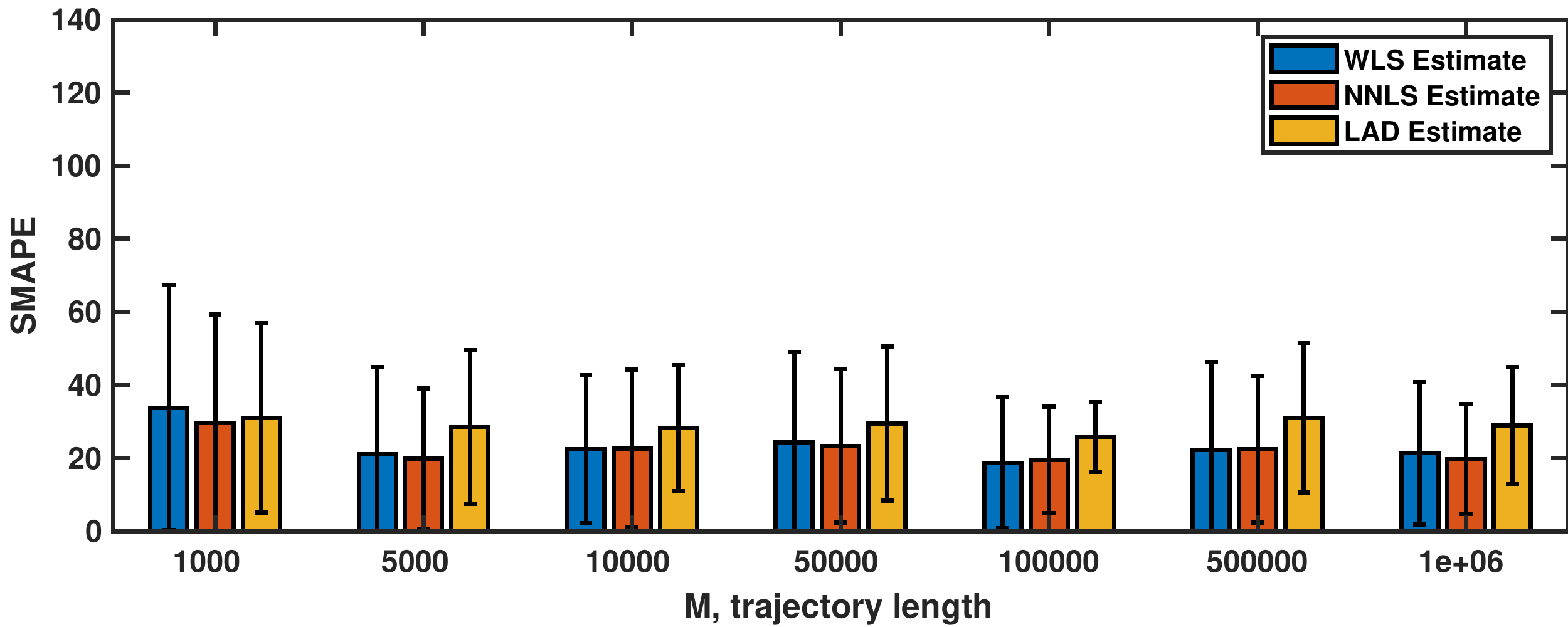}  
  \caption{SMAPE of Healing Rate $\mu$}
  \label{fig:118buscontact_SMAPE_mu}
\end{subfigure}
\hfill

\begin{subfigure}{.5\textwidth}
  \centering
  \includegraphics[width=0.95\linewidth]{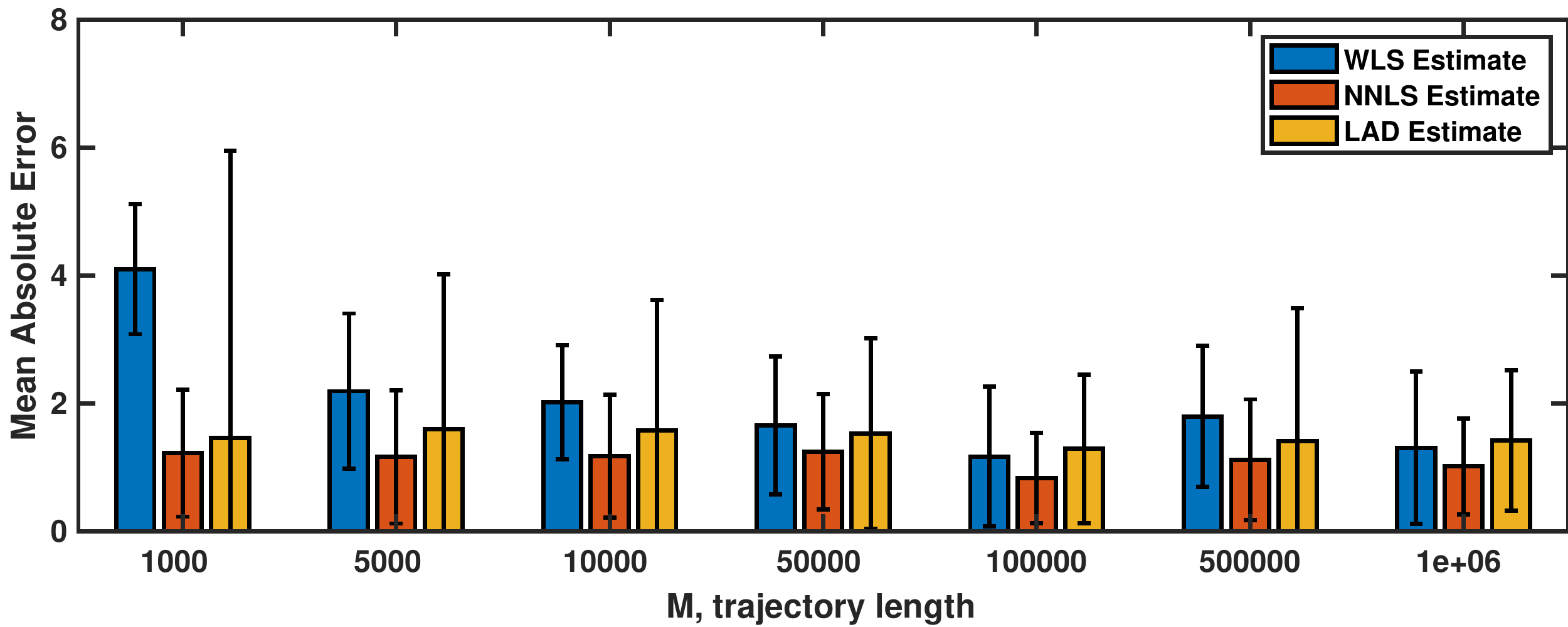}  
  \caption{MAE of Exogeneous Infection Rate $\beta$}
  \label{fig:ER100SMAPE}
\end{subfigure}
\begin{subfigure}{.5\textwidth}
  \centering
  \includegraphics[width=0.95\linewidth]{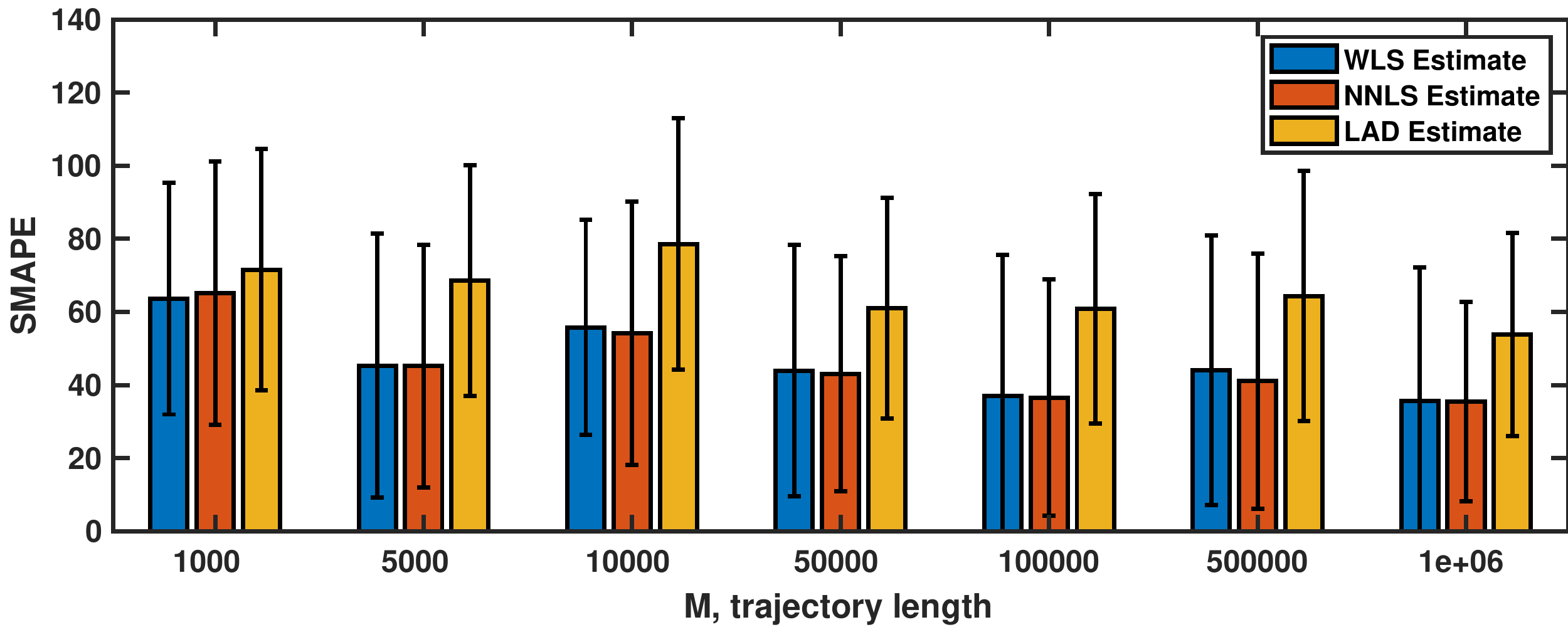}  
  \caption{SMAPE of Exogeneous Infection Rate $\beta$}
  \label{fig:ER100SMAPE}
\end{subfigure}
\hfill

\begin{subfigure}{.5\textwidth}
  \centering
  \includegraphics[width=0.95\linewidth]{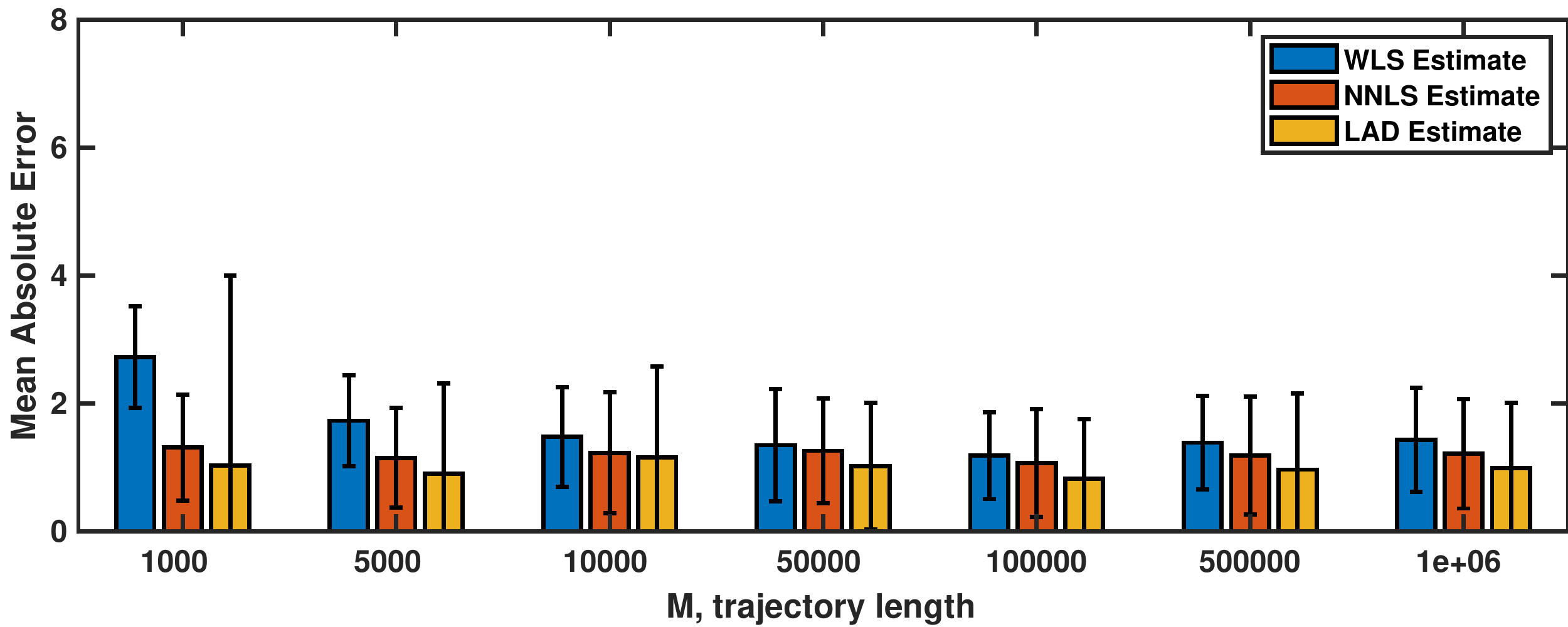}  
  \caption{MAE of Endogeneous Infection Rate $\delta$}
  \label{fig:ER100SMAPE}
\end{subfigure}
\begin{subfigure}{.5\textwidth}
  \centering
  \includegraphics[width=0.95\linewidth]{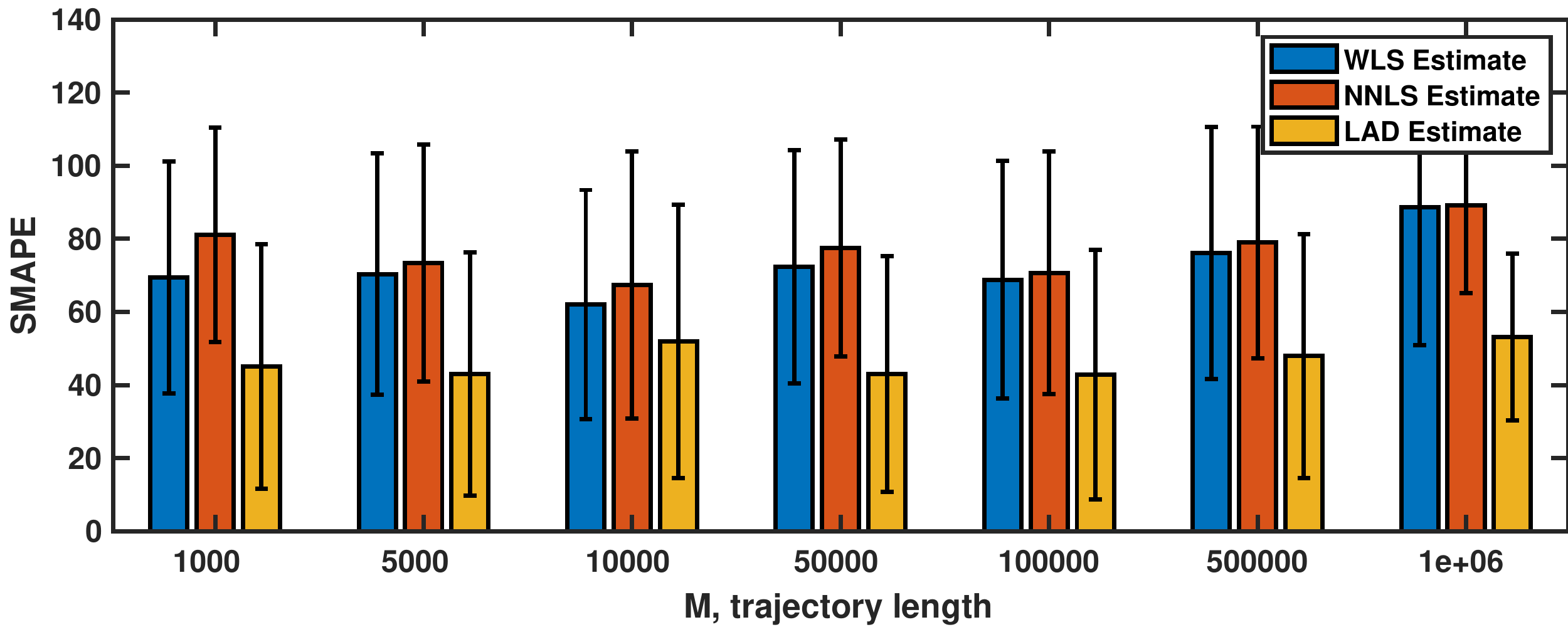}  
  \caption{SMAPE of Endogeneous Infection Rate $\delta$}
  \label{fig:ER100SMAPE}
\end{subfigure}
\caption{Contact Process: Error vs. Trajectory Length ($M$) for 100-Bus Power Grid}
\label{fig:contactpowerresult}
\end{figure*}

%
%


\subsection{Parameter Estimation Performance for Reversible Contact Processes}

\subsubsection{Reversible Contact Processes: Synthetic Erd\H{o}s-R\'{e}nyi (ER) Graphs}
We used the same ER graphs and generation process as in the previous section. The contact process dynamic is replaced by the reversible contact process. Note that from Section~\ref{sec:rev}, the transition rates of the reversible contact process are not directly a linear function of the healing and infection rates. Instead,  it is linear in higher dimensions. As a results, there are more parameters, $\theta = [\mu, \beta, \beta\delta, \beta\delta^2, \ldots \beta\delta^{d_{\max}}]^T$ need to be estimated via Algorithm~\ref{alg} than compared to the contact process. While we can directly obtain estimates of $\mu$ and $\beta$, an additional step is need to find $\delta$ from $ \beta\delta, \beta\delta^2, \ldots \beta\delta^{d_{\max}}$. In contrast to the contact process, the estimate of the endogenous infection rate $\delta$ is poorer than the exogeneous infection rate $\beta$.

Figure~\ref{fig:revER100result} shows the the MAE and SMAPE error for the healing rate, endogeneous infection rate, and exogeneous infection rate. The estimates for the reversible contact process have larger error compared to the contact process. Note that the absolute error is shown in \emph{log} scale. The standard deviation of the estimates can also be very large.

\begin{figure*}[ht]
\begin{subfigure}{0.5\textwidth}
  \centering
  \includegraphics[width=0.95\linewidth]{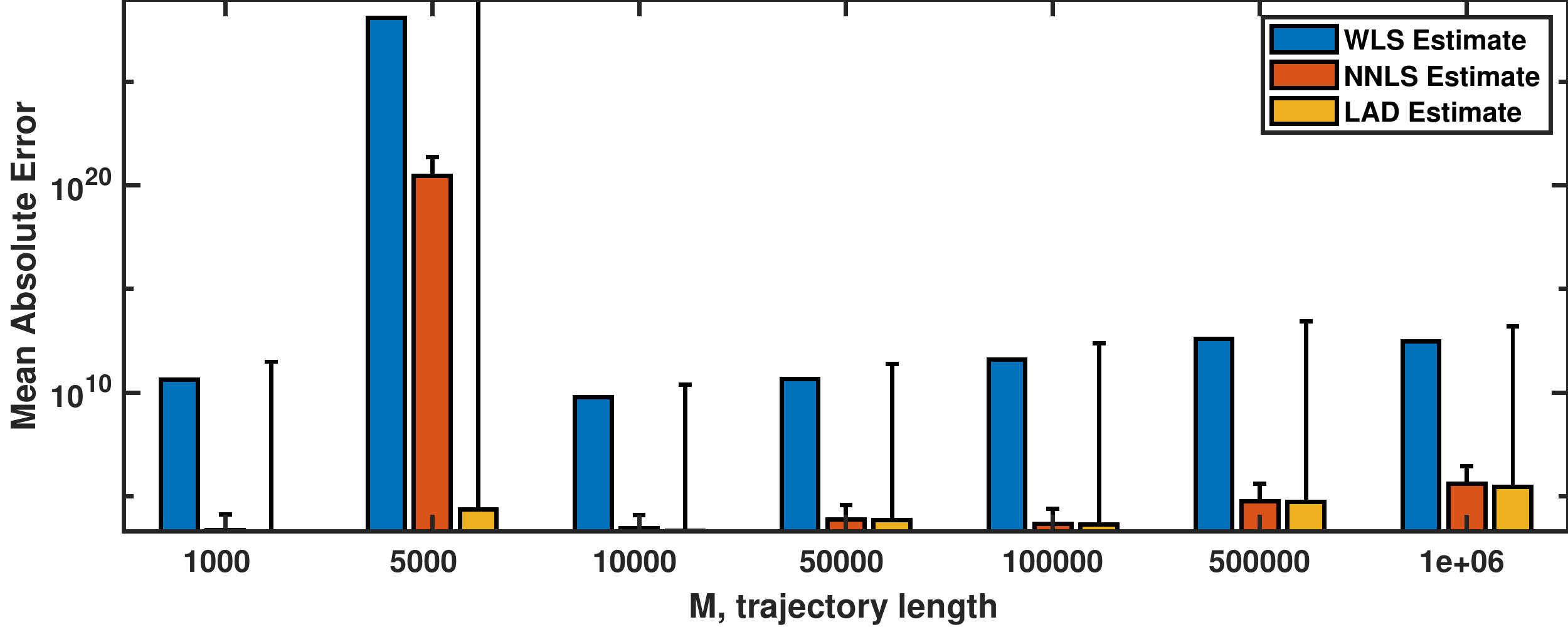}  
  \caption{Mean Absolute Error of Healing Rate $\mu$}
  \label{fig:ER100MAE}
\end{subfigure}
\begin{subfigure}{0.5\textwidth}
  \centering
  \includegraphics[width=0.95\linewidth]{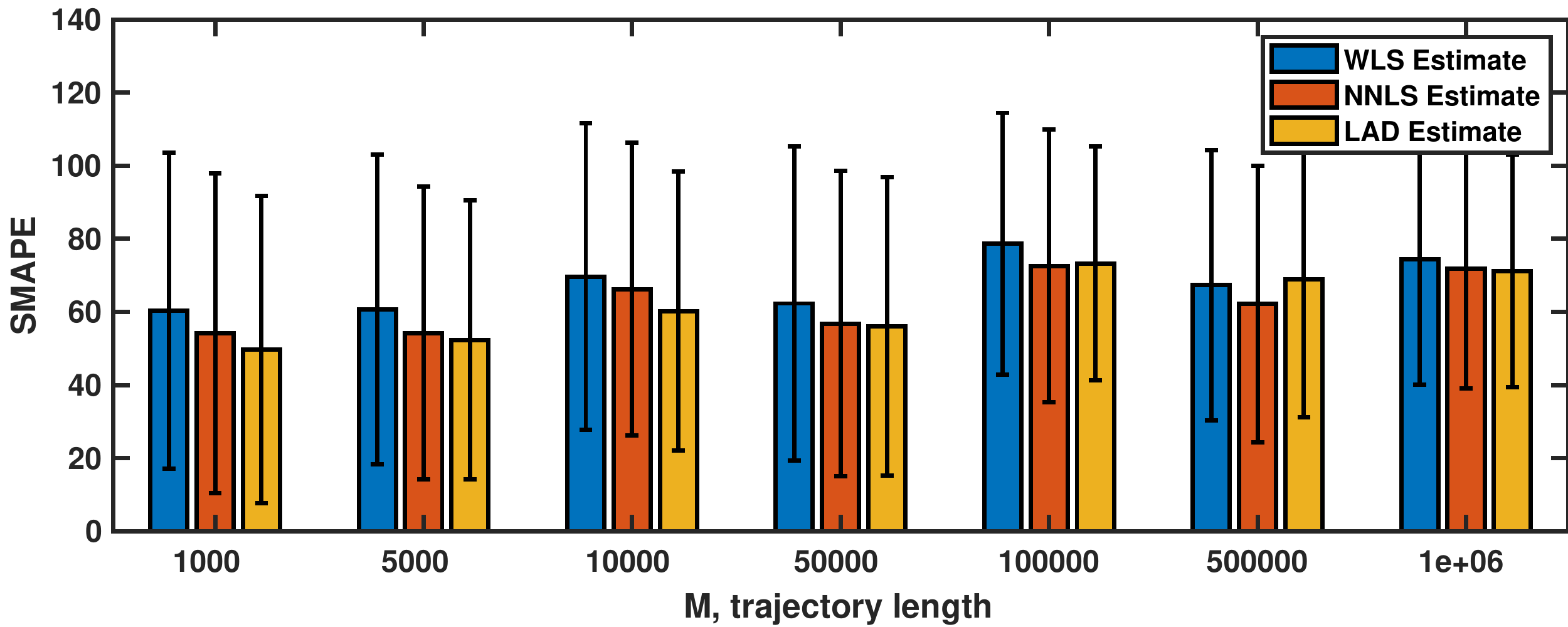}  
  \caption{SMAPE of Healing Rate $\mu$}
  \label{fig:ER100MAE}
\end{subfigure}
\hfill

\begin{subfigure}{0.5\textwidth}
  \centering
  \includegraphics[width=0.95\linewidth]{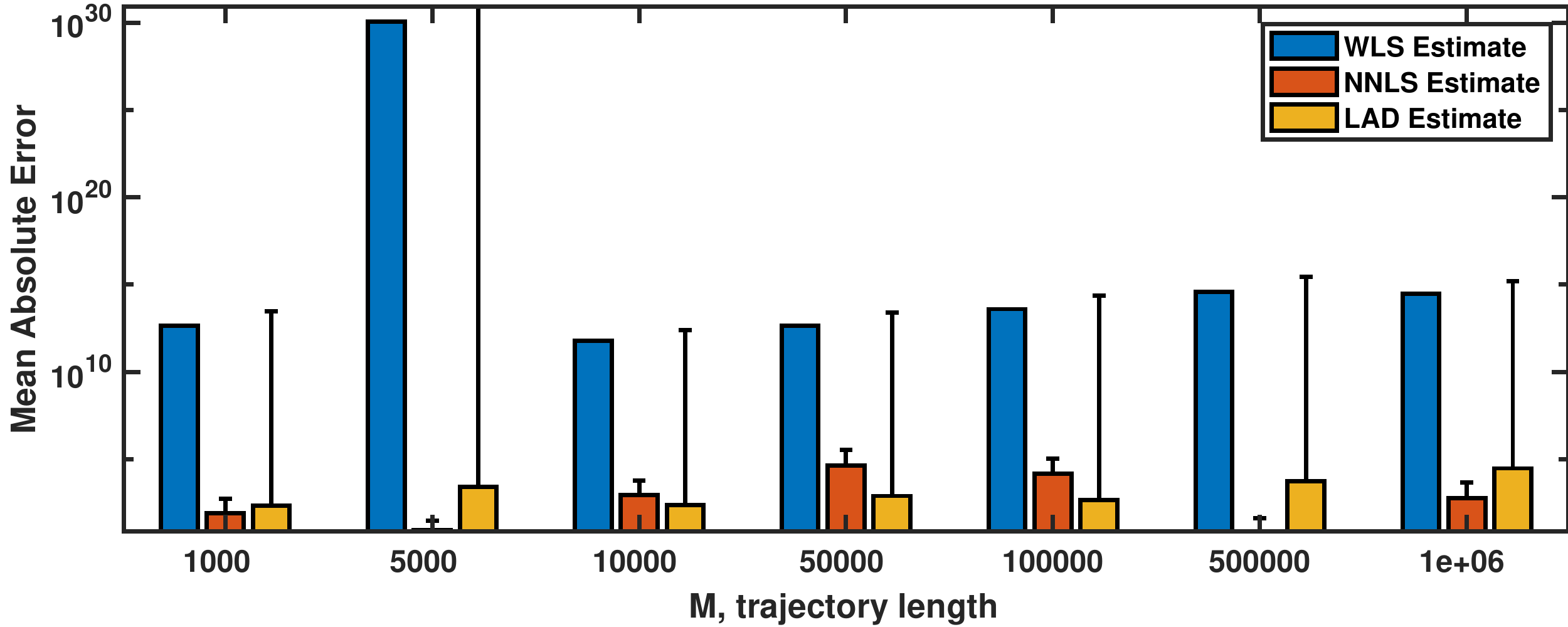}  
  \caption{Mean Absolute Error of Exogeneous Infection Rate $\beta$}
  \label{fig:ER100SMAPE}
\end{subfigure}
\begin{subfigure}{.5\textwidth}
  \centering
  \includegraphics[width=0.95\linewidth]{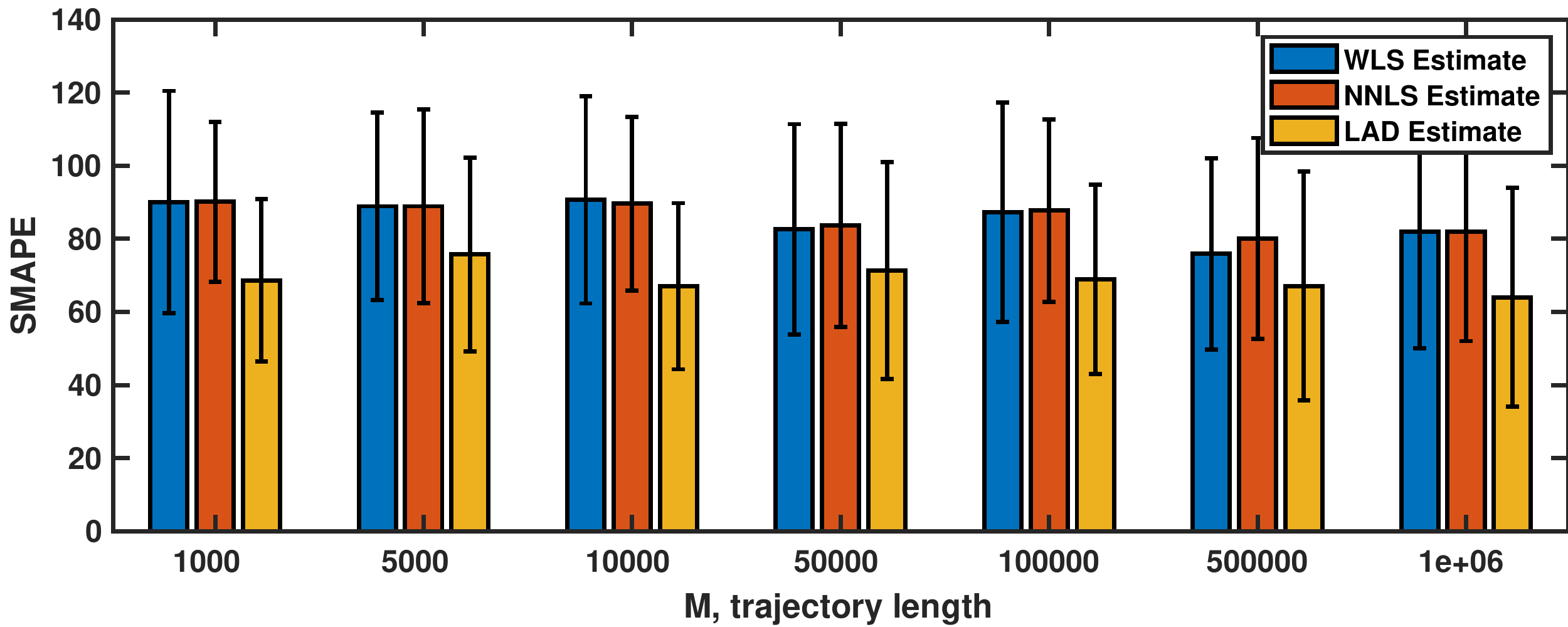}  
  \caption{SMAPE of Exogeneous Infection Rate $\beta$}
  \label{fig:ER100SMAPE}
\end{subfigure}
\hfill

\begin{subfigure}{.5\textwidth}
  \centering
  \includegraphics[width=0.95\linewidth]{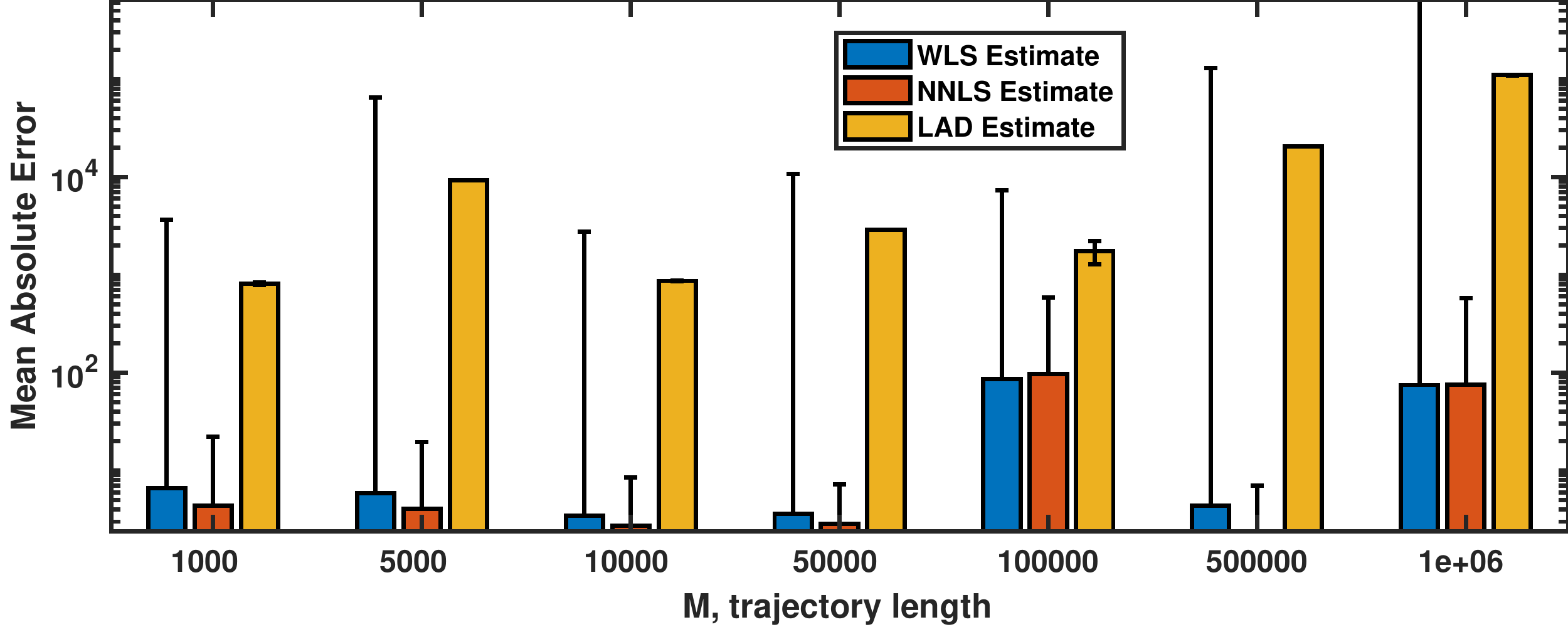}  
  \caption{Mean Absolute Error of Endogeneous Infection Rate $\delta$}
  \label{fig:ER100SMAPE}
\end{subfigure}
\begin{subfigure}{.5\textwidth}
  \centering
  \includegraphics[width=0.95\linewidth]{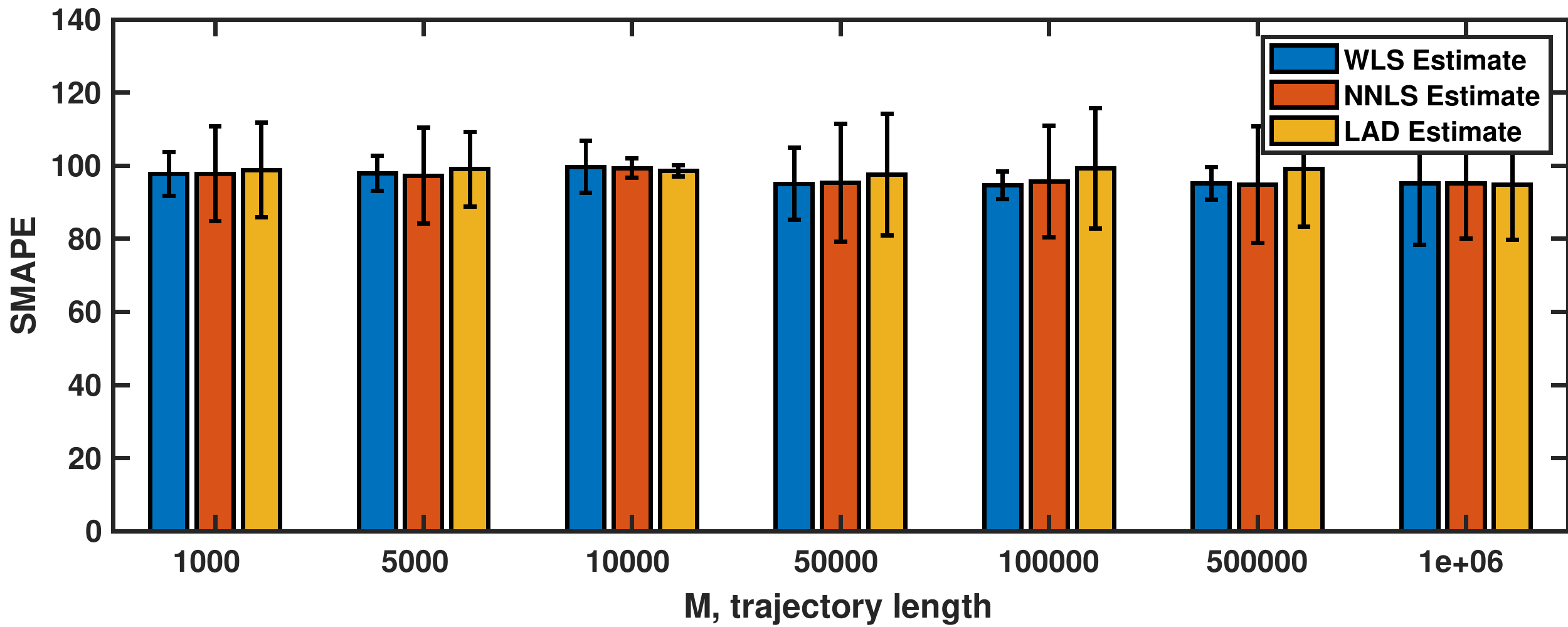}  
  \caption{SMAPE of Endogeneous Infection Rate $\delta$}
  \label{fig:ER100SMAPE}
\end{subfigure}
\caption{Reversible Contact Process: Error vs. Trajectory Length ($M$) for Erd\H{o}s-R\'{e}nyi Graphs}
\label{fig:revER100result}
\end{figure*}

\subsubsection{\textbf{Reversible Contact Processes: 118-bus Power Grid Graph}}

Figure~\ref{fig:revbus118result} shows the mean absolute error and SMAPE from 50 independent trajectories $\Sigma$ from the reversible contact process with random initialization and dynamics parameters. The errors are much lower for $\mu, \beta, \delta$ compared to experiments using randomly generated ER graphs. This may be because the 110-bus network has very low average degree, and the maximum degree is limited to $9$. The NNLS estimator gives slightly lower error than the other estimators. We can also see that the error of the reversible contact process is in general larger than the contact process.

\begin{figure*}[ht]

\begin{subfigure}{.5\textwidth}
  \centering
  \includegraphics[width=0.95\linewidth]{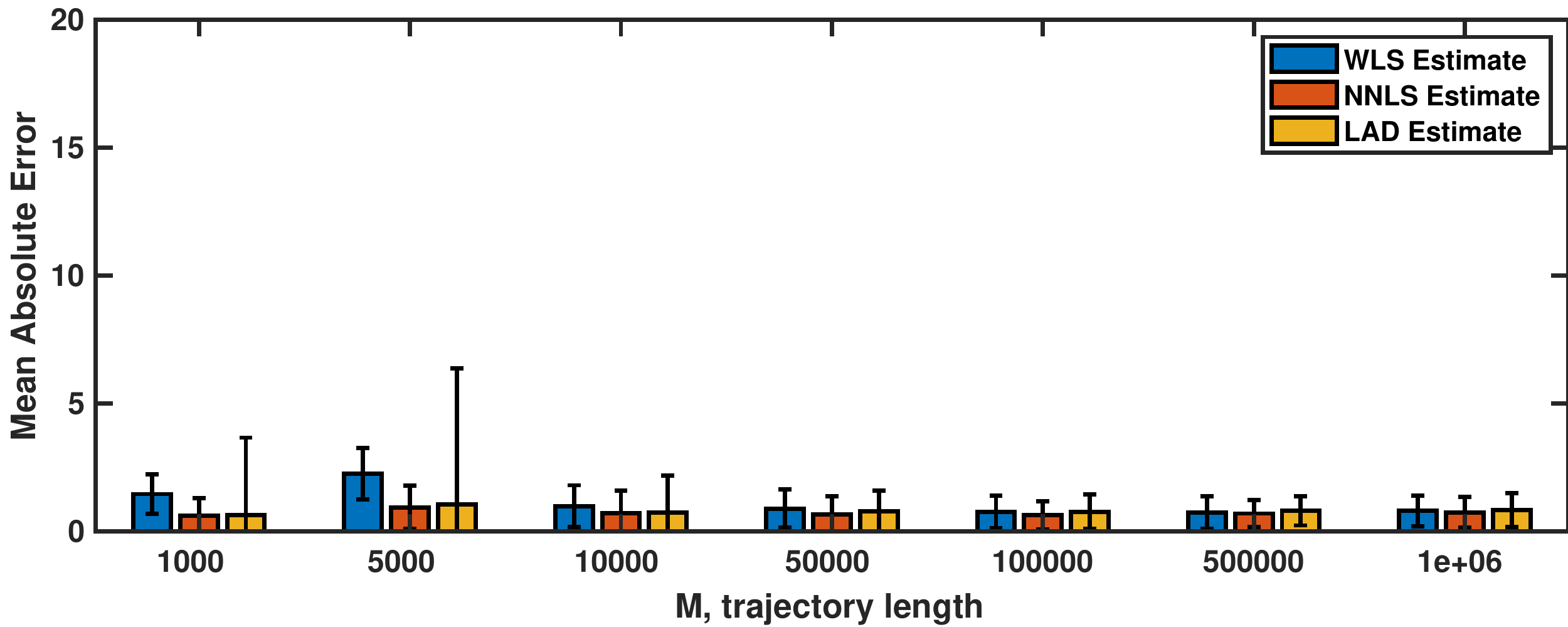}  
  \caption{Mean Absolute Error of Healing Rate $\mu$}
  \label{fig:118BusMAE}
\end{subfigure}
\begin{subfigure}{.5\textwidth}
  \centering
  \includegraphics[width=0.95\linewidth]{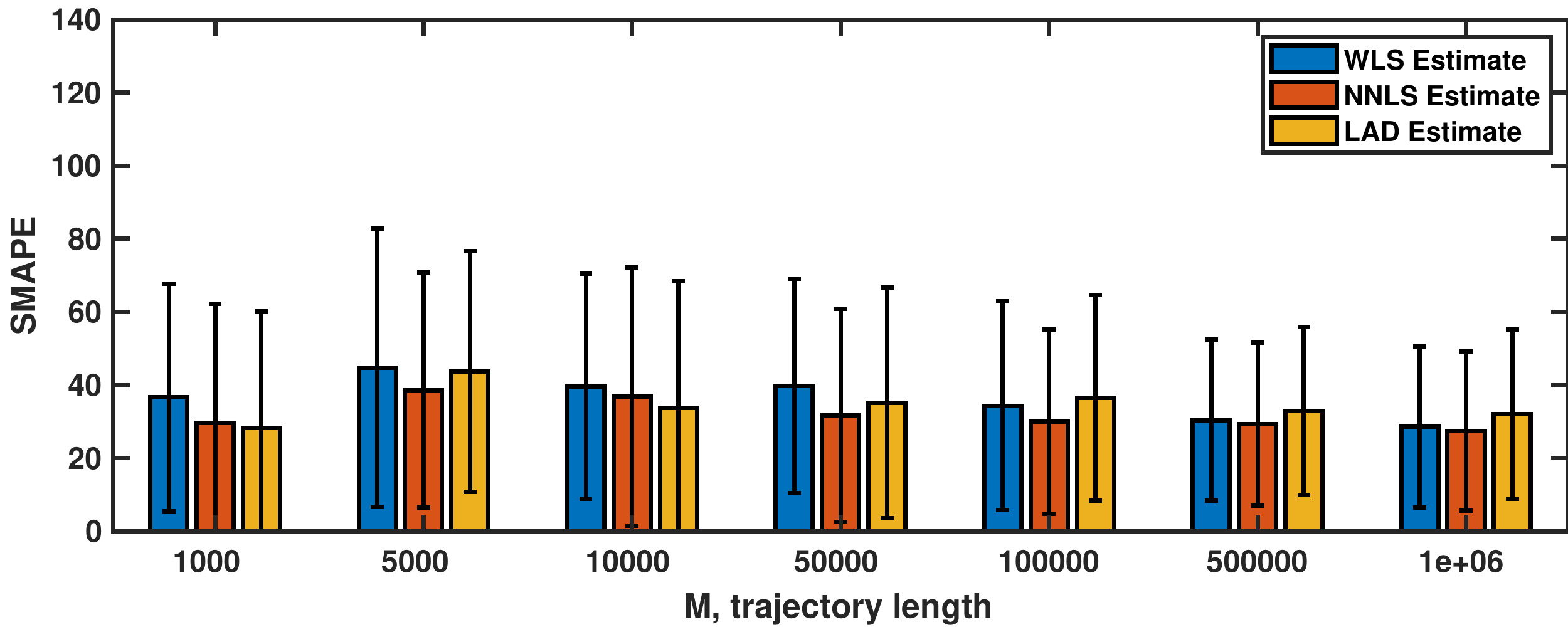}  
  \caption{SMAPE of Healing Rate $\mu$}
  \label{fig:ER100MAE}
\end{subfigure}
\hfill

\begin{subfigure}{.5\textwidth}
  \centering
  \includegraphics[width=0.95\linewidth]{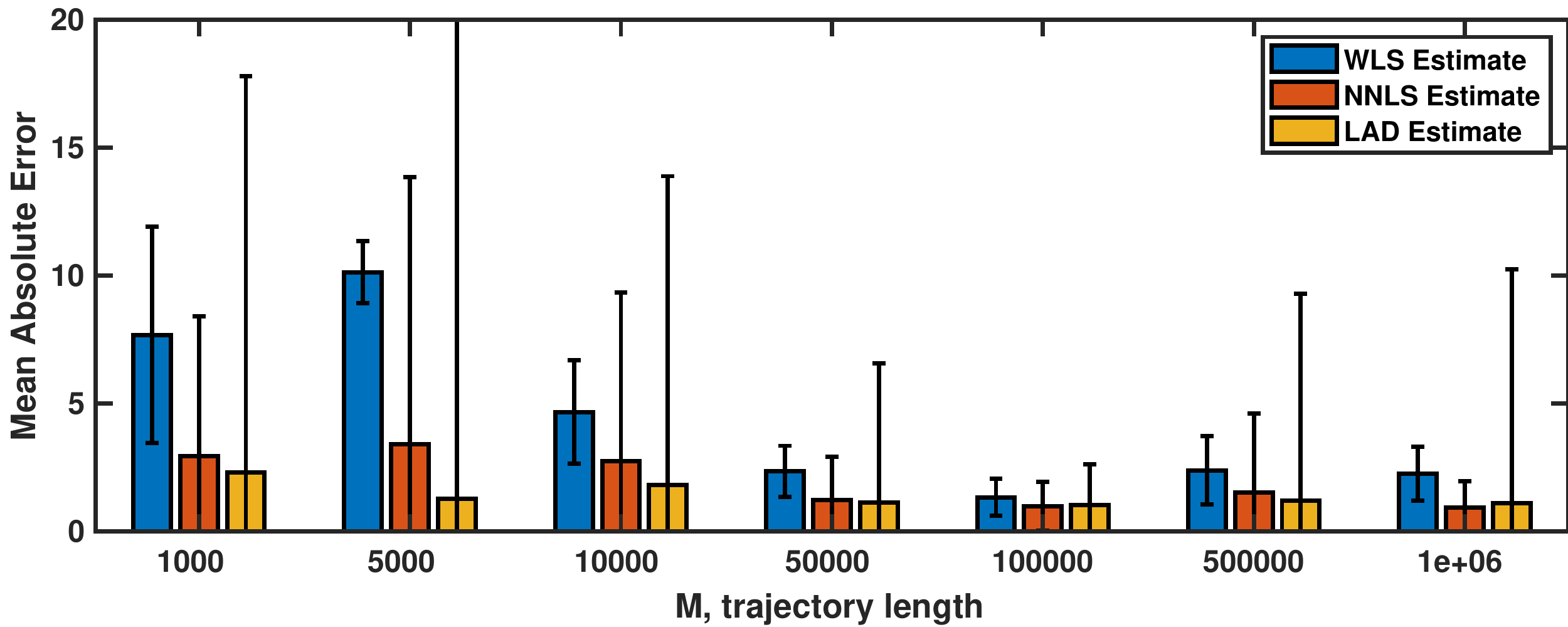}  
  \caption{Mean Absolute Error of Exogeneous Infection Rate $\beta$}
  \label{fig:ER100SMAPE}
\end{subfigure}
\begin{subfigure}{.5\textwidth}
  \centering
  \includegraphics[width=0.95\linewidth]{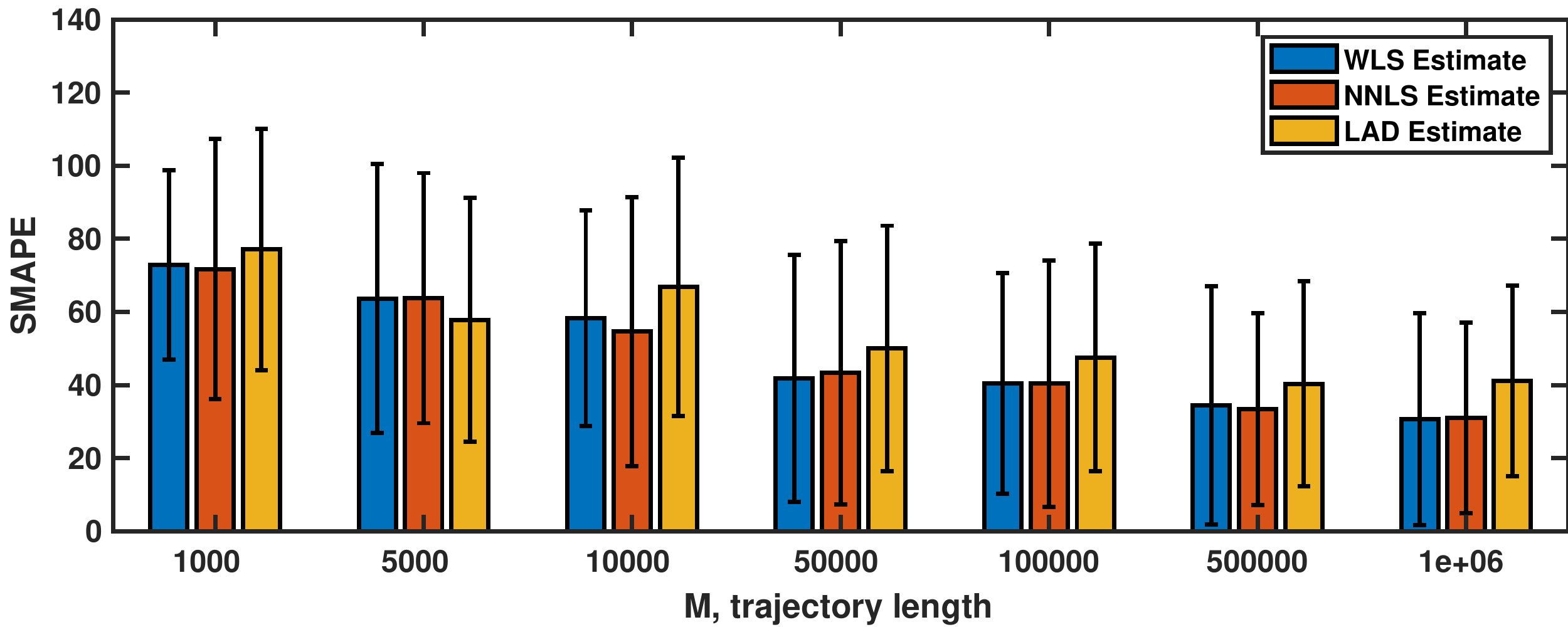}  
  \caption{SMAPE of Exogeneous Infection Rate $\beta$}
  \label{fig:ER100SMAPE}
\end{subfigure}
\hfill

\begin{subfigure}{.5\textwidth}
  \centering
  \includegraphics[width=0.95\linewidth]{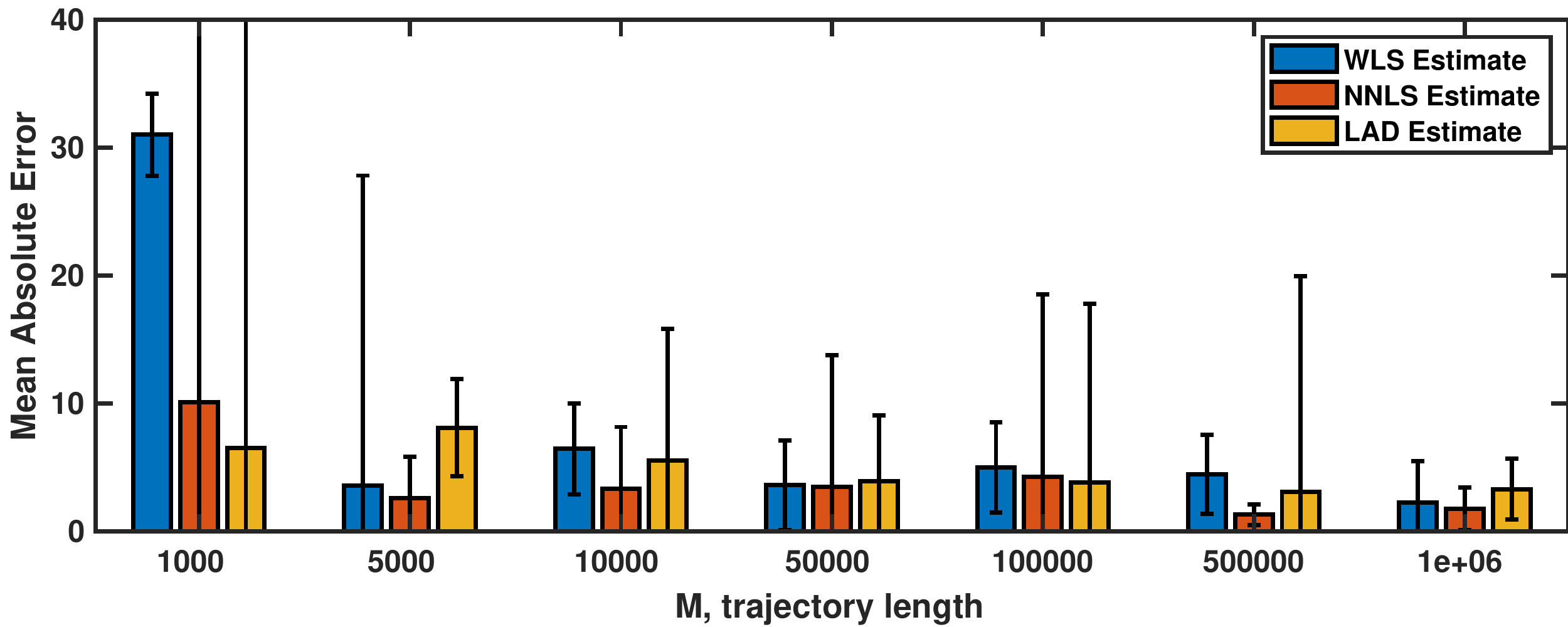}  
  \caption{Mean Absolute Error of Endogeneous Infection Rate $\delta$}
  \label{fig:ER100SMAPE}
\end{subfigure}
\begin{subfigure}{.5\textwidth}
  \centering
  \includegraphics[width=0.95\linewidth]{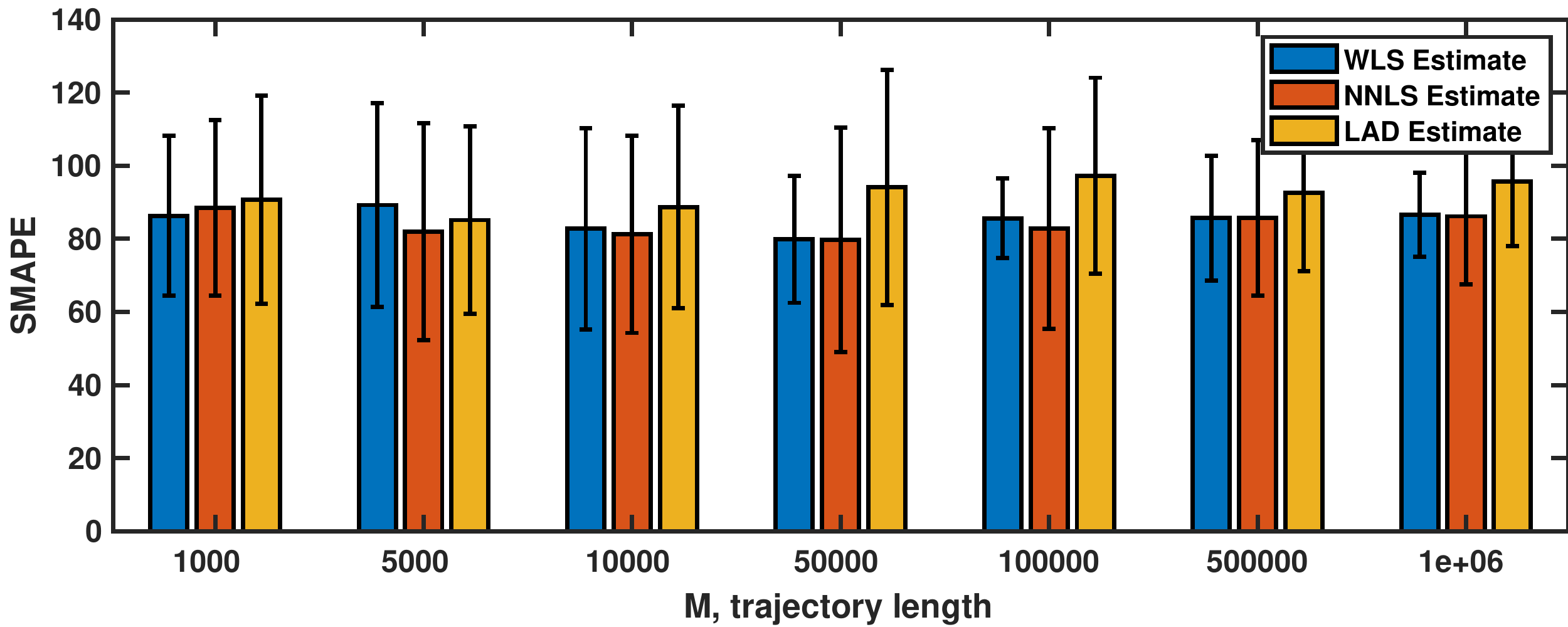}  
  \caption{SMAPE of Endogeneous Infection Rate $\delta$}
  \label{fig:ER100SMAPE}
\end{subfigure}
\caption{Reversible Contact Process: Error vs. Trajectory Length ($M$) for 100-Bus Power Grid}
\label{fig:revbus118result}
\end{figure*}

\section{Conclusion}\label{Sec:Con}

We studied continuous-time Markov processes whose transition rates are linear functions of a smaller set of dynamics parameters; our goal is to develop a method to estimate these parameters from a single continuously-observed trajectory (i.e., no missing observations). We defined the concept of holding classes, which are Markov states in a continuous-time Markov process that have the same expected holding time. It can be shown that for network-based epidemics models, it is more efficient to consider holding classes than individual Markov states. A much smaller system of equation can be formulated to estimate the dynamics parameters via weighted least squares methods. We showed with numerical experiments on synthetic and real-world graphs that reasonable estimates of dynamics parameters can be obtained even the length of the trajectory is many magnitudes smaller than the dimensionality of the Markov process. In all criteria, the estimation method works better for the contact process than the reversible contact process.


There are several directions for future work. It is not clear the dependency between network structure and the number of holding classes. Better upperbound may be derived, especially the study of asymptotic behavior of $K$ as $N \to \infty$. There are many more issues to address with estimation.  Realistic observations would have noise and/or missing values (i.e., discrete-time observations). We believe that estimation results can be improved using generalized least squares as the errors for different holding classes are not (realistically) independent. The challenge is on how to estimate the (non-diagonal) covariance matrix from a single trajectory.

\ifCLASSOPTIONcompsoc
  \section*{Acknowledgments}
\else
  \section*{Acknowledgment}
\fi
This work was supported in part by the National Science Foundation AI Institute in Dynamic Systems (Grant No. 2112085).

\bibliography{library}

\begin{thebibliography}{10}

\bibitem{harris1974contact}
Theodore~E Harris,
\newblock ``Contact interactions on a lattice,''
\newblock {\em The Annals of Probability}, pp. 969--988, 1974.

\bibitem{kermack1927contribution}
William~O. Kermack and Anderson~G. McKendrick,
\newblock ``A contribution to the mathematical theory of epidemics,''
\newblock {\em Proceedings of the Royal Society of London. Series A, Containing
  papers of a mathematical and physical character}, vol. 115, no. 772, pp.
  700--721, 1927.

\bibitem{anderson1992infectious}
Roy~M. Anderson and Robert~M. May,
\newblock {\em Infectious Diseases of Humans: Dynamics and Control},
\newblock Oxford university press, 1992.

\bibitem{dietz2002daniel}
Klaus Dietz and J.A.P. Heesterbeek,
\newblock ``Daniel bernoulli's epidemiological model revisited,''
\newblock {\em Mathematical Biosciences}, vol. 180, no. 1-2, pp. 1--21, 2002.

\bibitem{Pan2014}
Jiafeng Pan, Alison Gray, David Greenhalgh, and Xuerong Mao,
\newblock ``{Parameter estimation for the stochastic SIS epidemic model},''
\newblock {\em Statistical Inference for Stochastic Processes}, vol. 17, no. 1,
  pp. 75--98, 2014.

\bibitem{Zheng2017}
Zhiwei Zheng, Huisheng Shu, Xiu Kan, Yingyi Fang, and Xin Zhang,
\newblock ``{Parameter Estimation for the Continuous Time Stochastic Logistic
  Diffusion Model},''
\newblock {\em Open Journal of Statistics}, vol. 07, no. 06, pp. 1039--1052,
  2017.

\bibitem{o2002tutorial}
Philip~D. O`Neill,
\newblock ``A tutorial introduction to bayesian inference for stochastic
  epidemic models using markov chain monte carlo methods,''
\newblock {\em Mathematical Biosciences}, vol. 180, no. 1-2, pp. 103--114,
  2002.

\bibitem{kypraios2017tutorial}
Theodore Kypraios, Peter Neal, and Dennis Prangle,
\newblock ``A tutorial introduction to bayesian inference for stochastic
  epidemic models using approximate bayesian computation,''
\newblock {\em Mathematical Biosciences}, vol. 287, pp. 42--53, 2017.

\bibitem{Dutta2018}
Ritabrata Dutta, Antonietta Mira, and Jukka~Pekka Onnela,
\newblock ``{Bayesian inference of spreading processes on networks},''
\newblock {\em Proceedings of the Royal Society A: Mathematical, Physical and
  Engineering Sciences}, vol. 474, no. 2215, pp. 1--24, 2018.

\bibitem{Pastor-Satorras2015a}
Romualdo Pastor-Satorras, Claudio Castellano, Piet {Van Mieghem}, and
  Alessandro Vespignani,
\newblock ``{Epidemic processes in complex networks},''
\newblock {\em Reviews of Modern Physics}, vol. 87, no. 3, pp. 1--62, 2015.

\bibitem{Pellis2015}
Lorenzo Pellis, Frank Ball, Shweta Bansal, Ken Eames, Thomas House, Valerie
  Isham, and Pieter Trapman,
\newblock ``{Eight challenges for network epidemic models},''
\newblock {\em Epidemics}, vol. 10, pp. 58--62, 2015.

\bibitem{gleeson2013binary}
James~P. Gleeson,
\newblock ``Binary-state dynamics on complex networks: Pair approximation and
  beyond,''
\newblock {\em Physical Review X}, vol. 3, no. 2, pp. 021004, 2013.

\bibitem{doi:10.1098/rsif.2005.0051}
Matt~J. Keeling and Ken~T.D. Eames,
\newblock ``Networks and epidemic models,''
\newblock {\em Journal of The Royal Society Interface}, vol. 2, no. 4, pp.
  295--307, 2005.

\bibitem{gomez2010discrete}
Sergio G{\'o}mez, Alexandre Arenas, J~Borge-Holthoefer, Sandro Meloni, and
  Yamir Moreno,
\newblock ``Discrete-time markov chain approach to contact-based disease
  spreading in complex networks,''
\newblock {\em EPL (Europhysics Letters)}, vol. 89, no. 3, pp. 38009, 2010.

\bibitem{Wang2017}
Wei Wang, Ming Tang, H.~{Eugene Stanley}, and Lidia~A. Braunstein,
\newblock ``{Unification of theoretical approaches for epidemic spreading on
  complex networks},''
\newblock {\em Reports on Progress in Physics}, vol. 80, no. 3, 2017.

\bibitem{Zhang2014}
June Zhang and Jos{\'{e}}~M.F. Moura,
\newblock ``{Diffusion in social networks as SIS epidemics: Beyond full mixing
  and complete graphs},''
\newblock {\em IEEE Journal on Selected Topics in Signal Processing}, vol. 8,
  no. 4, pp. 537--551, 2014.

\bibitem{Zhang2015}
June Zhang and Jos{\'{e}}~M.F. Moura,
\newblock ``{Role of subgraphs in epidemics over finite-size networks under the
  scaled SIS process},''
\newblock {\em Journal of Complex Networks}, vol. 3, no. 4, pp. 584--605, 03
  2015.

\bibitem{Zhang2017}
June Zhang and Jos{\'{e}}~M.F. Moura,
\newblock ``{Contact process with exogenous infection and the scaled SIS
  process},''
\newblock {\em Journal of Complex Networks}, vol. 5, no. 5, pp. 712--733, 2017.

\bibitem{VanMieghem2012}
Piet {Van Mieghem} and Eric Cator,
\newblock ``{Epidemics in networks with nodal self-infection and the epidemic
  threshold},''
\newblock {\em Physical Review E - Statistical, Nonlinear, and Soft Matter
  Physics}, vol. 86, no. 1, pp. 1--10, 2012.

\bibitem{liggett1985interacting}
Thomas~Milton Liggett and Thomas~M Liggett,
\newblock {\em Interacting particle systems}, vol.~2,
\newblock Springer, 1985.

\bibitem{griffeath1983binary}
David Griffeath,
\newblock ``The binary contact path process,''
\newblock {\em The Annals of Probability}, pp. 692--705, 1983.

\bibitem{glauber1963time}
Roy~J. Glauber,
\newblock ``Time-dependent statistics of the ising model,''
\newblock {\em Journal of mathematical physics}, vol. 4, no. 2, pp. 294--307,
  1963.

\bibitem{nowzari2016analysis}
Cameron Nowzari, Victor~M. Preciado, and George~J. Pappas,
\newblock ``Analysis and control of epidemics: A survey of spreading processes
  on complex networks,''
\newblock {\em IEEE Control Systems Magazine}, vol. 36, no. 1, pp. 26--46,
  2016.

\bibitem{Centola2007}
Damon Centola and Michael Macy,
\newblock ``Complex contagions and the weakness of long ties,''
\newblock {\em American Journal of Sociology}, vol. 113, no. 3, pp. 702--734,
  2007.

\bibitem{loukas2019stationary}
Andreas Loukas and Nathana{\"e}l Perraudin,
\newblock ``Stationary time-vertex signal processing,''
\newblock {\em EURASIP Journal on Advances in Signal processing}, vol. 2019,
  no. 1, pp. 1--19, 2019.

\bibitem{grassi2017time}
Francesco Grassi, Andreas Loukas, Nathana{\"e}l Perraudin, and Benjamin Ricaud,
\newblock ``A time-vertex signal processing framework: Scalable processing and
  meaningful representations for time-series on graphs,''
\newblock {\em IEEE Transactions on Signal Processing}, vol. 66, no. 3, pp.
  817--829, 2017.

\bibitem{nodelman2012continuous}
Uri Nodelman, Christian~R Shelton, and Daphne Koller,
\newblock ``Continuous time bayesian networks,''
\newblock {\em arXiv preprint arXiv:1301.0591}, 2012.

\bibitem{el2011continuous}
Tal El-Hay, Nir Friedman, Daphne Koller, and Raz Kupferman,
\newblock ``Continuous-time markov networks,''
\newblock {\em Reasoning about Structured Stochastic Systems in
  Continuous-Time}, p.~15, 2011.

\bibitem{mcgibbon2015efficient}
Robert~T. McGibbon and Vijay~S. Pande,
\newblock ``Efficient maximum likelihood parameterization of continuous-time
  markov processes,''
\newblock {\em The Journal of Chemical Physics}, vol. 143, no. 3, pp. 034109,
  2015.

\bibitem{MSMpaper}
Brooke~E. Husic and Vijay~S. Pande,
\newblock ``Markov state models: From an art to a science,''
\newblock {\em Journal of the American Chemical Society}, vol. 140, no. 7, pp.
  2386--2396, 02 2018.

\bibitem{norris1998markov}
James~R Norris,
\newblock {\em Markov Chains},
\newblock Number~2. Cambridge university press, 1998.

\bibitem{Kelly}
Frank~P Kelly,
\newblock {\em Reversibility and Stochastic Networks},
\newblock Cambridge University Press, 2011.

\bibitem{PhysRevE.86.016116}
Piet Van~Mieghem and Eric Cator,
\newblock ``Epidemics in networks with nodal self-infection and the epidemic
  threshold,''
\newblock {\em Phys. Rev. E}, vol. 86, pp. 016116, Jul 2012.

\bibitem{JZhang}
June Zhang and Jos{\'e}~MF Moura,
\newblock ``Cascading edge failures: A dynamic network process,''
\newblock {\em IEEE Transactions on Network Science and Engineering}, vol. 5,
  no. 4, pp. 288--300, 2017.

\bibitem{flores1986pragmatic}
Benito~E. Flores,
\newblock ``A pragmatic view of accuracy measurement in forecasting,''
\newblock {\em Omega}, vol. 14, no. 2, pp. 93--98, 1986.

\bibitem{christie1993power}
R.~Christie,
\newblock ``Power systems test case archive: 118 bus power flow test case,''
\newblock {\em University of Washington, Department of Electrical
  Engineering,[Online], Available: www. ee. washington.
  edu/research/pstca/pf118/pg tca118bus. htm}, 1993.

\end{thebibliography}

%

\begin{IEEEbiography}[{\includegraphics[width=1in,height=1.25in,clip,keepaspectratio]{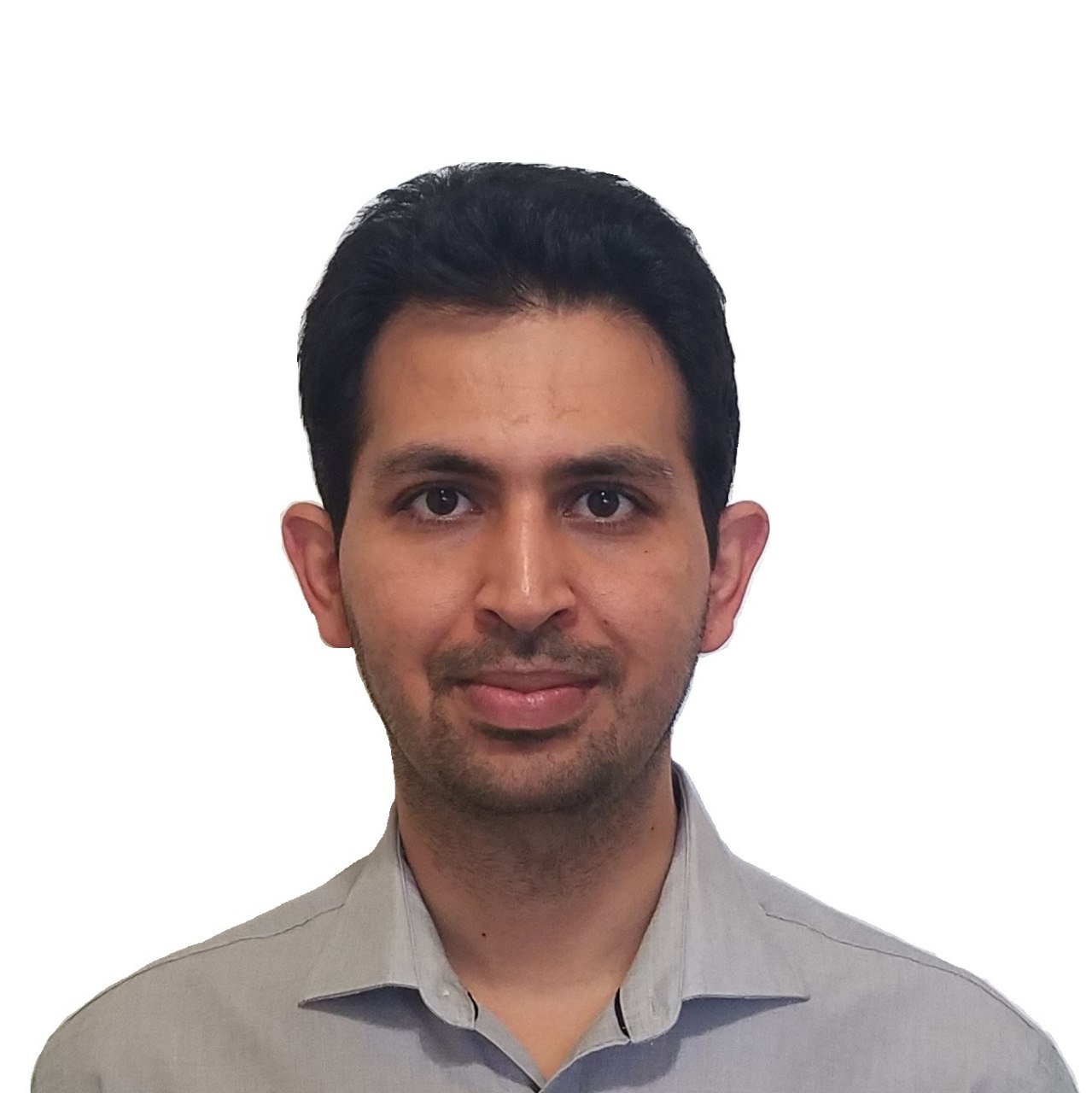}}]{Seyyed A. Fatemi} (Ph.D. U. of Hawaii Electrical Engineering) is a former post-doctoral researcher at the University of Hawai'i at M\={a}noa specializing in artificial intelligence, stochastic optimization, reinforcement learning, probabilistic and  machine learning modeling, statistical learning, estimation, prediction, time series analysis. His research interests include areas such as statistical signal processing, network science, big data and artificial intelligence, renewable energy and smart grid.
\end{IEEEbiography}

\begin{IEEEbiography}[{\includegraphics[width=1in,height=1.25in,clip,keepaspectratio]{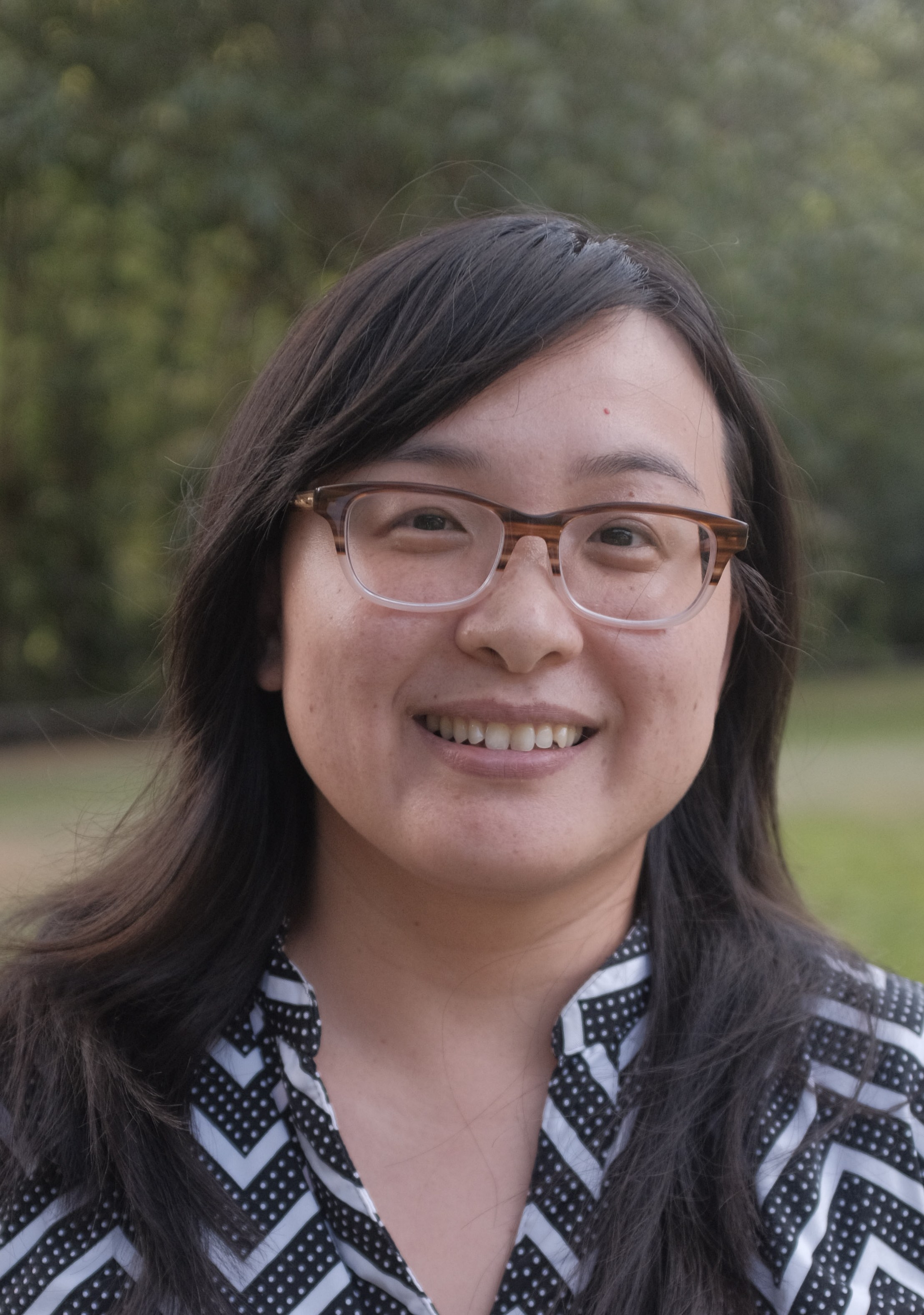}}]{June Zhang} is an assistant professor at the University of Hawai'i at M\={a}noa in the Department of Electrical and Computer Engineering. She received her B.S. with Highest Honor in Electrical and Computer Engineering from the Georgia Institute of Technology and M.S. in Electrical and Computer Engineering from Stanford University. She received a Ph.D. in Electrical and Computer Engineering from Carnegie Mellon University in December 2015. She is a co-PI of the NSF AI Institute in Dynamic Systems. Her research interests are graph dynamical systems, graph coding, representation learning, anomaly detection, graph signal processing, and natural language processing. 
\end{IEEEbiography}

\end{document}